\documentclass[envcountsame,envcountsect]{llncs}
\usepackage{ifpdf}
\ifpdf
\usepackage{pdfsync}
\fi
\usepackage[tbtags]{amsmath}
\usepackage{amssymb,cite,times}
\usepackage{amsfonts}
\usepackage{mathrsfs}
\usepackage{enumerate}
\usepackage{gastex}
\usepackage{graphics}
\usepackage{xcolor}
\usepackage{latexsym}
\usepackage{ifthen}
\usepackage{url}

\newcommand{\A}{\mathcal A}
\newcommand{\acc}{\mathsf{acc}}
\newcommand{\B}{\mathcal B}
\newcommand{\Buchi}{B\"{u}chi }
\newcommand{\C}{\mathcal C}
\newcommand{\D}{\mathcal D}
\newcommand{\F}{\mathcal F}
\newcommand{\FO}{\ensuremath{\mathsf{FO}}}
\newcommand{\func}{\mathsf{func}}

\renewcommand{\H}{\mathcal H}

\newcommand{\Iso}{\mathsf{Iso}}
\newcommand{\N}{\mathbb N}
\newcommand{\K}{\mathcal K}
\newcommand{\Run}{\mathrm{Run}}
\renewcommand{\S}{\mathcal S}
\newcommand{\T}{\mathcal T}
\newcommand{\Ti}{{\mathcal T}^{\mathrm i}}
\newcommand{\unfold}{\mathrm{unfold}}

\newcommand{\REC}{\mathsf{REC}}
\newcommand{\SC}{\mathsf{SC}}

\begin{document}
\title{The Isomorphism Problem for $\omega$-Automatic Trees}
\author{Dietrich Kuske\inst{1} \and Jiamou Liu\inst{2} \and  Markus Lohrey\inst{2,}
\thanks{The second and third author are supported by the DFG  research project GELO.}
\institute{
 Laboratoire Bordelais de Recherche en Informatique
  (LaBRI), CNRS and Universit\'e Bordeaux I, Bordeaux, France
\and
Universit\"at
Leipzig, Institut f\"ur Informatik, Germany
\\
\email{kuske@labri.fr, liujiamou@gmail.com,
lohrey@informatik.uni-leipzig.de}}}

\maketitle
\pagestyle{plain}

\begin{abstract}
  The main result of this paper states that the isomorphism for
  $\omega$-au\-tomatic trees of finite height is at least has hard as
  second-order arithmetic and therefore not analytical. This
  strengthens a recent result by Hjorth, Khoussainov, Montalb\'an, and
  Nies \cite{HjorthKMN08} showing that the isomorphism problem for
  $\omega$-automatic structures is not in~$\Sigma^1_2$. Moreover,
  assuming the continuum hypothesis {\bf CH}, we can show that the isomorphism
  problem for $\omega$-automatic trees of finite height is
  recursively equivalent with second-order arithmetic.
  On the way to our main results, we show lower and upper bounds for the
  isomorphism problem for $\omega$-automatic trees of every finite height:
  (i) It is decidable ($\Pi^0_1$-complete, resp,) for height 1 (2,
  resp.), (ii)~$\Pi^1_1$-hard and in $\Pi^1_2$ for height 3, and
  (iii)~$\Pi^1_{n-3}$- and $\Sigma^1_{n-3}$-hard and in $\Pi^1_{2n-4}$
  (assuming {\bf CH}) for all $n\ge4$. All proofs are elementary and do not rely on
  theorems from set theory.
\end{abstract}

\section{Introduction}

A graph is computable if its domain is a computable set of natural
numbers and the edge relation is computable as well. Hence, one can
compute effectively in the graph.  On the other hand, practically
all other properties are undecidable for computable graphs (e.g.,
reachability, connectedness, and even the existence of isolated
nodes). In particular, the isomorphism problem is highly undecidable
in the sense that it is complete for $\Sigma^1_1$ (the first
existential level of the analytical hierarchy \cite{Odi89}); see e.g.\
\cite{CaKni06,GonKn02} for further investigations of the isomorphism
problem for computable structures. These algorithmic deficiencies have
motivated in computer science the study of more restricted classes of
finitely presented infinite graphs. For instance, pushdown graphs,
equational graphs, and prefix recognizable graphs have a decidable
monadic second-order theory and for the former two the isomorphism
problem is known to be decidable \cite{Cour89} (for prefix
recognizable graphs the status of the isomorphism problem seems to be
open).

Automatic graphs \cite{KhoN95} are in between prefix recognizable and
computable graphs.  In essence, a graph is automatic if the elements
of the universe can be represented as strings from a regular language
and the edge relation can be recognized by a finite state automaton
with several heads that proceed synchronously.  Automatic graphs (and
more general, automatic structures) received increasing interest over
the last years~\cite{BluG04,IsKhRu02,KhoNRS07,KhoRS05,Rub08}.  One of
the main motivations for investigating automatic graphs is that their
first-order theories can be decided uniformly (i.e., the input is an
automatic presentation and a first-order sentence).  On the other
hand, the isomorphism problem for automatic graphs is
$\Sigma^1_1$-complete \cite{KhoNRS07} and hence as complex as for
computable graphs (see \cite{KusL10} for the recursion theoretic
complexity of some more natural properties of automatic graphs).

In our recent paper \cite{KuLiLo10}, we studied the isomorphism
problem for restricted classes of automatic graphs. Among other
results, we proved that (i) the isomorphism problem for automatic
trees of height at most $n \geq 2$ is complete for the level
$\Pi^0_{2n-3}$ of the arithmetical hierarchy and (ii) that the
isomorphism problem for automatic trees of finite height is
recursively equivalent to true arithmetic.  In this paper, we extend
our techniques from \cite{KuLiLo10} to \emph{$\omega$-automatic
  trees}.  The class of $\omega$-automatic structures was introduced
in \cite{Blu99}, it generalizes automatic structures by replacing
ordinary finite automata by B\"uchi automata on $\omega$-words. In
this way, uncountable graphs can be specified. Some recent results on
$\omega$-automatic structures can be found in
\cite{KusL08,HjorthKMN08,KaRuBa08,Kus10}. On the logical side, many of
the positive results for automatic structures carry over to
$\omega$-automatic structures \cite{Blu99,KaRuBa08}. On the other hand, the
isomorphism problem of $\omega$-automatic structures is more
complicated than that of automatic structures (which is
$\Sigma^1_1$-complete). Hjorth et al.\ \cite{HjorthKMN08} constructed two
$\omega$-automatic structures for which the existence of an
isomorphism depends on the axioms
of set theory. Using Schoenfield's absoluteness theorem, they infer
that isomorphism of $\omega$-automatic structures does not belong
to~$\Sigma^1_2$. The extension of our elementary techniques from
\cite{KuLiLo10} to $\omega$-automatic trees allows us to show directly
(without a ``detour'' through set theory) that the isomorphism problem
for $\omega$-automatic trees of finite height is not analytical (i.e.,
does not belong to any of the levels $\Sigma^1_n$). For this, we prove
that the isomorphism problem for $\omega$-automatic trees of height $n
\geq 4$ is hard for both levels $\Sigma^1_{n-3}$ and $\Pi^1_{n-3}$ of
the analytical hierarchy (our proof is uniform in $n$).  A more
precise analysis moreover reveals at which height the complexity jump
for $\omega$-automatic trees occurs: For automatic as well as for
$\omega$-automatic trees of height 2, the isomorphism problem is
$\Pi^0_1$-complete and hence arithmetical.
But the isomorphism problem for $\omega$-automatic
trees of height 3 is hard for $\Pi^1_1$ (and therefore outside of the
arithmetical hierarchy) while the isomorphism problem for automatic
trees of height 3 is $\Pi^0_3$-complete~\cite{KuLiLo10}.
Our lower bounds for $\omega$-automatic trees even hold
for the smaller class of injectively $\omega$-automatic trees.

We prove our results by reductions from monadic second-order
(fragments of) number theory. The first step in the proof is a
normal form for analytical predicates. The basic idea of
the reduction then is that a subset $X \subseteq \N$ can be encoded by
an $\omega$-word~$w_X$ over $\{0,1\}$, where the $i$-th symbol is $1$
if and only if $i \in X$. The combination of this basic observation
with our techniques from \cite{KuLiLo10} allows us to encode monadic
second-order formulas over $(\N, +, \times)$ by $\omega$-automatic
trees of finite height.  This yields the lower bounds mentioned above.
We also give
an upper bound for the isomorphism problem: for $\omega$-automatic
trees of height $n$, the isomorphism problem belongs to $\Pi^1_{2n-4}$.
While the lower bound holds in the usual system {\bf ZFC} of
set theory, we can prove the upper bound only assuming in addition the
continuum hypothesis.
The precise recursion theoretic complexity of the isomorphism
problem for  $\omega$-automatic trees remains open, it might depend
on the underlying axioms for set theory.

\paragraph*{\bf Related work}
Results on isomorphism problems for various subclasses of automatic
structures can be found in \cite{KhoNRS07,KhoRS05,KuLiLo10,Rub04}.
Some completeness results for low levels of the analytical hierarchy
for decision problems on infinitary rational relations were shown
in \cite{Fin09}.

\section{Preliminaries}

Let $\N_+ = \{1,2,3,\ldots\}$.
With $\overline{x}$ we denote a tuple $(x_1,\ldots,x_m)$ of variables, whose
length $m$ does not matter.

\subsection{The analytical hierarchy} \label{sec:analytical.hierarchy}

In this paper we follow the definitions of the arithmetical and
analytical hierarchy from~\cite{Odi89}.
In order to avoid some  technical complications, it is useful to
exclude $0$ in the following, i.e., to consider subsets of $\N_+$.
In the following, $f_i$ ranges over unary functions on $\N_+$, $X_i$
over subsets of $\N_+$, and $u,x,y,z,x_i,\ldots$ over elements of $\N_+$.
The class
$\Sigma^0_n\subseteq 2^{\N_+}$ is the collection of all sets
$A\subseteq \N_+$ of the form
\begin{equation*}
    A = \{x\in \N_+ \mid (\N, +, \times) \models
   \exists y_1 \ \forall y_2 \cdots Q y_n: \varphi(x,y_1,\ldots,y_n) \},
\end{equation*}
where $Q = \forall$ (resp.\ $Q=\exists$) if $n$ is even (resp.\ odd)
and $\varphi$ is a quantifier-free formula over the signature
containing $+$ and $\times$.  The class $\Pi^0_n$ is the class of all
complements of $\Sigma^0_n$ sets. The classes $\Sigma^0_n, \Pi^0_n$
($n\geq 1$) make up the \emph{arithmetical hierarchy}.

The analytical hierarchy extends the arithmetical hierarchy and is
defined analogously using function quantifiers:  The class $\Sigma^1_n \subseteq
2^{\N_+}$ is the collection of all sets $A \subseteq \N_+$ of the form
\begin{equation*}
A = \{ x \in \N_+ \mid (\N, +, \times)
\models  \exists f_1 \ \forall f_2 \cdots Q f_n : \varphi(x,f_1,\ldots,f_n)\},
\end{equation*}
where $Q = \forall$ (resp.\ $Q = \exists$) if $n$ is even (resp.\ odd)
and $\varphi$ is a first-order formula over the signature containing
$+$, $\times$, and the functions $f_1, \ldots, f_n$.  The class
$\Pi^1_n$ is the class of all complements of $\Sigma^1_n$ sets.  The
classes $\Sigma^1_n, \Pi^1_n$ ($n \geq 1$) make up the
\emph{analytical hierarchy},
see Figure~\ref{fig:hierarchy}
for an inclusion diagram.
 The class of {\em analytical
  sets}\footnote{Here the notion of {\em analytical sets} is defined
  for sets of natural numbers and is not to be confused with the {\em
    analytic sets} studied in descriptive set theory \cite{Kech95}.} is exactly
$\bigcup_{n \geq 1} \Sigma^1_n$. 

\begin{figure}[t]
\begin{center}
\setlength{\unitlength}{.7mm}
\begin{picture}(130,25)(-10,0)
    \gasset{Nframe=n,AHnb=0,ExtNL=n}
    \gasset{Nfill=n,Nadjust=wh,Nadjustdist=1}    
    \node(r)(-5,10){$\bigcup_{n\geq 1} \Sigma^0_n$}
    \node(s1)(15,0){$\Sigma_1^1$}
    \node(p1)(15,20){$\Pi_1^1$}
    \node(s2)(60,0){$\Sigma_2^1$}
    \node(p2)(60,20){$\Pi_2^1$}
    \node(s3)(105,0){$\Sigma_3^1$}
    \node(p3)(105,20){$\Pi_3^1$}
    \drawedge(r,s1){}
    \drawedge(r,p1){}
    \drawedge(s1,p2){}
    \drawedge(s1,s2){}
    \drawedge(p1,p2){}
    \drawedge(p1,s2){}
    \drawedge(s2,p3){}
    \drawedge(s2,s3){}
    \drawedge(p2,p3){}
    \drawedge(p2,s3){}
    \put(110,10){$\ldots$}
\end{picture}
\end{center}
\caption{The analytical hierarchy}\label{fig:hierarchy}
\end{figure}
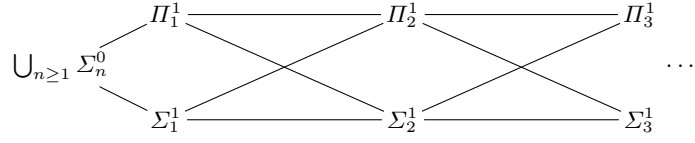

As usual in computability theory, a G\"odel numbering of all finite
objects of interest allows to quantify over, say, finite
automata as well. We will always assume such a numbering without mentioning it
explicitly.

\subsection{B\"{u}chi automata}

For details on B\"{u}chi automata, see
\cite{GrThWi02,PerP04,Tho97handbook}.  Let $\Gamma$ be a finite
alphabet. With $\Gamma^*$ we denote the set of all finite words over
the alphabet $\Gamma$. The set of all nonempty finite words is
$\Gamma^+$. An $\omega$-word over $\Gamma$ is an infinite sequence
$w=a_1a_2a_3\cdots$ with $a_i\in \Gamma$.  We set $w[i]=a_i$ for $i\in
\N_+$. The set of all $\omega$-words over $\Gamma$ is denoted by
$\Gamma^\omega$.

A (nondeterministic) B\"{u}chi automaton is a tuple $M=(Q, \Gamma,
\Delta, I, F)$, where $Q$ is a finite set of states, $I,F\subseteq Q$
are resp. the sets of initial and final states, and $\Delta\subseteq Q\times
\Gamma\times Q$ is the transition relation. If $\Gamma=\Sigma^n$ for
some alphabet $\Sigma$, then we refer to $M$ as an {\em
  $n$-dimensional B\"{u}chi automaton over $\Sigma$}. A {\em run} of
$M$ on an $\omega$-word $w=a_1a_2a_3\cdots$ is an $\omega$-word
$r=(q_1,a_1,q_2)(q_2,a_2,q_3)(q_3,a_3,q_4)\cdots\in\Delta^\omega$ such
that $q_1\in I$. The run $r$ is {\em accepting} if there exists a
final state from $F$ that occurs infinitely often in $r$. The language
$L(M)\subseteq \Gamma^\omega$ defined by $M$ is the set of all
$\omega$-words for which there exists an accepting run. An
$\omega$-language $L\subseteq \Gamma^\omega$ is {\em regular} if there
exists a B\"{u}chi automaton~$M$ with $L(M)=L$.  The class of all
regular $\omega$-languages is effectively closed under Boolean
operations and projections.

For $\omega$-words $w_1,\ldots, w_n\in \Gamma^\omega$, the
{\em convolution} $w_1\otimes w_2\otimes \cdots \otimes w_n\in (\Gamma^n)^\omega$ is defined by
\[
    w_1\otimes w_2 \otimes \cdots \otimes w_n =
(w_1[1],\ldots, w_n[1])(w_1[2],\ldots,w_n[2])(w_1[3],\ldots, w_n[3])\cdots .
\]
For $\overline{w}=(w_1,\ldots, w_n)$, we write $\otimes(\overline{w})$ for $w_1\otimes \cdots \otimes w_n$.

An $n$-ary relation $R\subseteq (\Gamma^\omega)^n$ is called
\emph{$\omega$-automatic} if the $\omega$-language $\otimes
R=\{\otimes(\overline{w}) \mid \overline{w} \in R\}$ is regular, i.e.,
it is accepted by some $n$-dimensional B\"{u}chi automaton. We denote
with $R(M) \subseteq (\Gamma^\omega)^n$ the relation defined by an
$n$-dimensional B\"{u}chi-automaton over the alphabet $\Gamma$.

To also define the convolution of finite words (and of finite words
with infinite words), we identify a finite word $u\in \Gamma^*$
with the $\omega$-word $u\diamond^\omega$, where $\diamond$ is a new
symbol. Then, for $u,v\in \Gamma^*, w\in \Gamma^\omega$, we write
$u\otimes v$ for the $\omega$-word $u\diamond^\omega \otimes
v\diamond^\omega$ and $u\otimes w$ (resp.\ $w\otimes u$) for
$u\diamond^\omega \otimes w$ (resp.\ $w\otimes u\diamond^\omega$).

In the following we describe some simple operations on \Buchi automata
that are used in this paper.
\begin{itemize}
\item Given two \Buchi automata $M_0=(Q_0,\Gamma,I_0,\Delta_0,F_0)$
  and $M_1=(Q_1,\Gamma,I_1,\Delta_1,F_1)$, we use $M_0\uplus M_1$ to
  denote the automaton obtained by taking the disjoint union of $M_0$
  and $M_1$. Note that for any word $u\in \Gamma^\omega$, the number
  of accepting runs of $M_0\uplus M_1$ on $u$ equals the sum of the
  numbers of accepting runs of $M_0$ and $M_1$ on $u$.
\item Let, again, $M_i=(Q_i,\Gamma,I_i,\Delta_i,F_i)$ for $i \in \{0,1\}$ be
  two B\"uchi automata. Then the intersection of their languages is
  accepted by the B\"uchi automaton $$M=(Q_0\times
  Q_1\times\{0,1\},\Gamma,I_0\times I_1\times\{0\},\Delta,F_0\times
  Q_1\times\{0\}),$$ where $((p_0,p_1,m),a,(q_0,q_1,n))\in\Delta$ if and
  only if
  \begin{itemize}
  \item $(p_0,a,q_0)\in\Delta_1$ and $(p_1,a,q_1)\in\Delta_1$, and
  \item  if $p_m \not\in F_m$ then $n=m$ and 
  if $p_m \in F_m$ then $n=1-m$.
  \end{itemize}
  Hence the runs of $M$ on the $\omega$-word $u$ consist of a run of
  $M_0$ and of $M_1$ on $u$. The ``flag'' $m\in\{0,1\}$ in
  $(p_0,p_1,m)$ signals that the automaton waits for an accepting
  state of $M_m$. As soon as such an accepting state is seen, the flag
  toggles its value. Hence accepting runs of $M$ correspond to pairs
  of accepting runs of $M_0$ and of $M_1$. Therefore, the number of
  accepting runs of $M$ on $u$ equals the product of the numbers of
  accepting runs of $M_0$ and of $M_1$ on $u$. This construction is
  known as the flag or Choueka construction
  (cf.~\cite{Cho74,Tho90,PerP04}).
 \item Let $\Sigma$ be an alphabet and $M=(Q,\Gamma,I,\Delta,F)$ be a
   B\"uchi automaton.  We use $\Sigma^\omega \otimes M$ to denote the
   automaton obtained from $M$ by expanding the alphabet to
   $\Sigma\times \Gamma$:
  \[
      \Sigma^\omega\otimes M = (Q, \Sigma\times \Gamma, I, \Delta', F),
  \]
  where $\Delta' = \{(p, (\sigma,a), q)\mid (p,a,q)\in \Delta,
  \sigma\in \Sigma\}$.  Note that $L(\Sigma^\omega\otimes M) =
  \Sigma^\omega \otimes L(\A)$.  
\end{itemize}

\subsection{$\omega$-automatic structures}\label{sec:prelim.structure}

A {\em signature} is a finite set $\tau$ of relational symbols
together with an arity $n_S\in\N_+$ for every relational symbol $S\in
\tau$. A {\em $\tau$-structure} is a tuple $\A=(A, (S^\A)_{S\in
  \tau})$, where $A$ is a set (the {\em universe} of $\A$) and
$S^{\A}\subseteq A^{n_S}$. When the context is clear, we denote $S^\A$
with $S$, and we write $a\in \A$ for $a\in A$.  Let $E\subseteq A^2$
be an equivalence relation on~$A$. Then $E$ is a \emph{congruence}
on $\A$ if
$(u_1,v_1),\ldots,(u_{n_S},v_{n_S}) \in E$ and $(u_1,\ldots, u_{n_S})
\in S$ imply $(v_1,\ldots,v_{n_S})\in S$ for all $S\in\tau$. Then
the {\em quotient structure} $\A/E$ can be defined:
\begin{itemize}
\item The universe of $\A/E$ is the set of all $E$-equivalence classes
  $[u]$ for $u\in A$.
\item The interpretation of $S\in \tau$ is the relation
  $\{([u_1],\ldots, [u_{n_S}]) \mid (u_1,\ldots,u_{n_S})\in S\}$.
\end{itemize}

\begin{definition}\label{Def-omega-automatic-presentation}
  An {\em $\omega$-automatic presentation} over the signature $\tau$
  is a tuple
  \[
     P=(\Gamma, M, M_\equiv, (M_S)_{S\in \tau})
  \]
  with the following properties:
  \begin{itemize}
  \item $\Gamma$ is a finite alphabet
  \item $M$ is a B\"uchi automaton over the alphabet $\Gamma$.
  \item For every $S\in \tau$, $M_S$ is an $n_S$-dimensional \Buchi
    automaton over the alphabet $\Gamma$.
  \item $M_{\equiv}$ is a 2-dimensional B\"uchi automaton over the
    alphabet $\Gamma$ such that $R(M_\equiv)$ is a
    congruence relation on $(L(M),(R(M_S))_{S\in\tau})$.
  \end{itemize}
  The {\em $\tau$-structure defined} by the $\omega$-automatic
  presentation $P$ is the quotient structure
  \[
    \S(P)=(L(M), (R(M_S))_{S\in\tau})/R(M_\equiv) \,.
  \]
\end{definition}
If $R(M_\equiv)$ is the identity relation on $\Gamma^\omega$, then $P$
is called {\em injective}. A structure $\A$ is {\em (injectively)
  $\omega$-automatic} if there is an (injectively) $\omega$-automatic
presentation $P$ with $\A \cong \S(P)$.  In ~\cite{HjorthKMN08} it
was shown that there exist
$\omega$-automatic structures that are not injectively
$\omega$-automatic.  We simplify our statements by
saying ``given/compute an (injectively) $\omega$-automatic structure
$\A$'' for ``given/compute an (injectively) $\omega$-automatic
presentation $P$ of a structure~$\S(P)\cong\A$''.  {\em Automatic
  structures} \cite{KhoN95} are defined analogously to
$\omega$-automatic structures, but instead of B\"uchi automata
ordinary finite automata over finite words are used.  For this, one
has to pad shorter strings with the padding symbol $\diamond$ when
defining the convolution of finite strings.  More details on
$\omega$-automatic structures can be found in
\cite{BluG04,HjorthKMN08,KaRuBa08}. In particular, a countable
structure is $\omega$-automatic if and only if it is
automatic~\cite{KaRuBa08}.

Let $\FO[\exists^{\aleph_0},\exists^{2^{\aleph_0}}]$ be first-order
logic extended by the quantifiers $\exists^{\kappa} x \ldots$ ($\kappa
\in \{\aleph_0,2^{\aleph_0}\}$) saying that there exist exactly
$\kappa$ many $x$ satisfying $\ldots$.  The following theorem lays out
the main motivation for investigating $\omega$-automatic structures.

\begin{theorem}[\cite{Blu99,KaRuBa08}]\label{thm:extFOaut}
  From an $\omega$-automatic presentation
  $$P=(\Gamma,M,M_\equiv,(M_S)_{S\in\tau})$$ and a formula
  $\varphi(\overline{x}) \in
  \FO[\exists^{\aleph_0},\exists^{2^{\aleph_0}}]$ in the signature
  $\tau$ with $n$ free variables, one can compute a B\"uchi automaton
  for the relation
  \[
     \{ (a_1,\ldots, a_n) \in L(M)^n \mid \S(P)\models\varphi([a_1],[a_2],\dots,[a_n])\}\,.
  \]
  In particular, the $\FO[\exists^{\aleph_0},\exists^{2^{\aleph_0}}]$
  theory of any $\omega$-automatic structure $\A$ is (uniformly)
  decidable.
\end{theorem}
\begin{definition}
  Let $\K$ be a class of $\omega$-automatic presentations. The
  \emph{isomorphism problem} $\Iso(\K)$ is the set of pairs
  $(P_1,P_2)\in\K^2$ of $\omega$-automatic presentations from $\K$
  with $\S(P_1)\cong\S(P_2)$.
\end{definition}
If $\S_1$ and $\S_2$ are two structures over the same signature, we
write $\S_1 \uplus \S_2$ for the disjoint union of the two
structures. We use $\S^{\kappa}$ to denote the disjoint union of
$\kappa$ many copies of the structure $\S$, where $\kappa$ is any
cardinal.

The disjoint union as well as the countable or uncountable power of an
automatic structure are effectively automatic, again. In this paper,
we will only need this property (in a more explicite form) for
injectively $\omega$-automatic structures.

\begin{lemma}\label{lem:multiply}
  Let $P_i=(\Gamma,M^i,M_\equiv^i,(M_S^i)_{S\in\tau})$ be injective
  $\omega$-automatic presentations of structures~$\S_i$ for $i\in\{1,2\}$.
  One can effectively construct injectively $\omega$-automatic copies
  of $\S_1\uplus\S_2$, $\S_1^{\aleph_0}$, and $\S_1^{2^{\aleph_0}}$
  such that
  \begin{itemize}
  \item The universe of the injectively $\omega$-automatic copy $\S$
    of $\S_1\uplus\S_2$ equals $L(M^1)\cup L(M^2)$ and the relations
    are given by $S^{\S}=R(M^1_S)\cup R(M^2_S)$ provided $L(M^1)$ and
    $L(M^2)$ are disjoint.
  \item The universe of the injectively $\omega$-automatic copy $\S$
    of $\S_1^{\aleph_0}$ is $\$^*\otimes L(M^1)$ where $\$$ is a
    fresh symbol. For $i\in\N$, the restriction of $\S$ to
    $\{\$^i\}\otimes L(M^1)$ forms a copy of $\S_1$.
  \item The universe of the injectively $\omega$-automatic copy $\S$
    of $\S_1^{2^{\aleph_0}}$ is $\{\$_1,\$_2\}^\omega\otimes L(M^1)$
    where $\$_1$ and $\$_2$ are fresh symbols. For
    $w\in\{\$_1,\$_2\}^\omega$, the restriction of $\S$ to $\{w\}\otimes
    L(M^1)$ forms a copy of $\S_1$.
  \end{itemize}
\end{lemma}

\subsection{Trees}

A {\em forest} is a partial order $F=(V,\le)$ such that for every $x
\in V$, the set $\{y\mid y\leq x\}$ of ancestors of $x$ is finite and
linearly ordered by~$\le$.  The {\em level} of a node $x\in V$ is
$|\{y \mid y < x \}| \in \N$. The {\em height} of $F$ is the supremum
of the levels of all nodes in $V$; it may be infinite.
Note that a forest of infinite height can be well-founded, i.e., all its paths are finite.
In this paper we only deal with forests of {\em finite height}. 
For all $u\in V$, $F(u)$
denotes the restriction of $F$ to the set $\{v\in V\mid u\le v\}$ of
successors of $u$. We will speak of the \emph{subtree rooted at $u$}.
A \emph{tree} is a forest that has a minimal element, called the
\emph{root}. For a forest $F$ and $r$ not belonging to the domain of
$F$, we denote with $r\circ F$ the tree that results from adding $r$
to $F$ as a new root. The \emph{edge relation}~$E$ of the forest $F$
is the set of pairs $(u,v)\in V^2$ such that $u$ is the largest
element in $\{ x \mid x < v \}$. 
Note that  a forest $F = (V,\le)$ of finite height is (injectively) $\omega$-automatic
if and only if the graph $(V,E)$ (where $E$ is the edge relation of $E$)
is (injectively) $\omega$-automatic, since each of these structures is first-order
interpretable in the other structure. This does not hold for trees of infinite 
height. For any node $u\in V$, we use $E(u)$
to denote the set of children (or immediate successors) of $u$.

We use $\T_n$ (resp.\ $\Ti_n$) to denote the class of (injectively)
$\omega$-automatic presentations of trees of height at most~$n$. Note
that it is decidable whether a given $\omega$-automatic presentation
$P$ belongs to $\T_n$ and $\Ti_n$, resp., since the class of trees of
height at most $n$ can be axiomatized in first-order logic.

\section{$\omega$-automatic trees of height 1 and 2}\label{sec:h2}

For $\omega$-automatic trees of height 2 we need the following result:
\begin{theorem} [\cite{KaRuBa08}]\label{thm:counting}
  Let $\A$ be an $\omega$-automatic structure and let
  $\varphi(x_1, \ldots,x_n,y)$ be a formula of
  $\FO[\exists^{\aleph_0},\exists^{2^{\aleph_0}}]$.  Then, for all
  $a_1,\ldots,a_n \in \A$, the cardinality of the set $\{ b \in \A \mid
  \A \models \varphi(a_1,\ldots,a_n,b) \}$ belongs to $\N\cup
  \{\aleph_0, 2^{\aleph_0}\}$.
\end{theorem}

\begin{theorem}\label{T-height-2} The following holds:
  \begin{itemize}
  \item
  The isomorphism problem $\Iso(\T_1)$ for $\omega$-automatic trees of
  height 1 is decidable.
  \item
  There exists a tree $U$ such that $\{P\in\Ti_2\mid \S(P)\cong U\}$
  is $\Pi^0_1$-hard. The isomorphism problems $\Iso(\T_2)$ and
  $\Iso(\Ti_2)$ for (injectively) $\omega$-automatic trees of height 2
  are $\Pi^0_1$-complete.
  \end{itemize}
\end{theorem}

\begin{proof}
  Two trees of height 1 are isomorphic if and only if they have the
  same size.  By Theorem~\ref{thm:counting}, the number of elements in
  an $\omega$-automatic tree $\S(P)$ with $P\in \T_1$ is either finite, $\aleph_0$
  or $2^{\aleph_0}$ and the exact size can be computed using
  Theorem~\ref{thm:extFOaut} (by checking successively validity of the
  sentences $\exists^\kappa x:x=x$ for
  $\kappa\in\N\cup\{\aleph_0,2^{\aleph_0}\}$\footnote{Where
    $\exists^n x:\varphi(x)$ for $n\in\N$ is shorthand for the obvious
    first-order formula expressing that there are exactly $n$ elements
    satisfying $\varphi$.}).

  Now, let us take two trees $T_1$ and $T_2$ of height 2 and let $E_i$
  be the edge relation of $T_i$ and $r_i$ its root.  For $i \in
  \{1,2\}$ and a cardinal $\lambda$ let $\kappa_{\lambda,i}$ be the
  cardinality of the set of all $u \in E_i(r_i)$ such that $|E_i(u)| =
  \lambda$.  Then $T_1 \cong T_2$ if and only if $\kappa_{\lambda,1} =
  \kappa_{\lambda,2}$ for any cardinal $\lambda$. Now assume that
  $T_1$ and $T_2$ are both $\omega$-automatic.  By
  Theorem~\ref{thm:counting}, for all $i \in \{1,2\}$ and every $u \in
  E_i(r_i)$ we have $|E_i(u)| \in \N\cup\{\aleph_0,2^{\aleph_0}\}$.
  Moreover, again by Theorem~\ref{thm:counting}, every cardinal
  $\kappa_{\lambda,1}$ ($\lambda \in \N\cup\{\aleph_0,2^{\aleph_0}\}$)
  belongs to $\N\cup\{\aleph_0,2^{\aleph_0}\}$ as well.  Hence, $T_1
  \cong T_2$ if and only if, for all
  $\kappa,\lambda\in\N\cup\{\aleph_0,2^{\aleph_0}\}$:
  \begin{align*}
    \phantom{\text{ if and only if }} &
    T_1\models\exists^\kappa x : ((r_1, x) \in E \wedge
    \exists^\lambda y : (x,y) \in E) \\\text{ if and only if }&
    T_2\models \exists^\kappa x : ((r_2, x) \in E \wedge
    \exists^\lambda y : (x,y) \in E )\,.
  \end{align*}
  By Theorem~\ref{thm:extFOaut}, this equivalence is decidable for all
  $\kappa,\lambda$. Since it has to hold for all $\kappa,\lambda$, the
  isomorphism of two $\omega$-automatic trees of height 2 is
  expressible by a $\Pi^0_1$-statement.  Hardness for $\Pi^0_1$
  follows from the corresponding result on automatic trees of
  height 2.\qed
\end{proof}

\section{A normal form for analytical sets}

To prove our lower bound for the isomorphism problem of
$\omega$-automatic trees of height $n\ge3$, we will use the following
normal form of analytical sets. A formula of the form $x \in X$ or $x
\not\in X$ is called a {\em set constraint}. The constructions in the
proof of the following lemma are standard. 

\begin{proposition}\label{prp:normal.form}
  For every odd $($resp. even$)$ $n\in \N_+$ and every $\Pi^1_n$
  $($resp. $\Sigma^1_n)$ relation $A\subseteq \N_+^r$, there exist
  polynomials $p_i,q_i\in \N[\overline{x},y,\overline{z}]$ and
  disjunctions $\psi_i$ $(1\leq i\leq \ell)$ of set constraints (on
  the set variables $X_1, \ldots, X_n$ and individual variables
  $\overline{x} ,y,\overline z$) such that $\overline x\in A$ if and
  only if
  \[
 Q_1 X_1\ Q_2 X_2\cdots
       Q_n X_n\ \exists y\  \forall\overline z : \bigwedge_{i=1}^\ell
       p_i(\overline{x} ,y,\overline z) \neq
       q_i(\overline{x} ,y,\overline z) \vee  \psi_i(\overline{x} ,y,\overline z,X_1,\ldots,X_n),
  \]
  where $Q_1,Q_2,\ldots,Q_n$ are alternating quantifiers
  with~$Q_n=\forall$.
\end{proposition}

\begin{proof} 
  For notational simplicity, we present the proof only for the case
  when $n$ is odd. The other case can be proved in a similar way by
  just adding an existential quantification $\exists X_0$ at the
  beginning. We will write $\Sigma_m(\SC,\REC)$ for the set of
  $\Sigma_m$-formulas over set constraints and recursive predicates,
  $\Pi_m(\SC,\REC)$ is to be understood similarly and
  $B\Sigma_m(\SC,\REC)$ is the set of boolean combinations of formulas
  from $\Sigma_m(\SC,\REC)$. With $C_k : \N_+^k \to \N_k$ we will
  denote some computable bijection.

  Fix an odd number $n$. It is well known that every
  $\Pi^1_n$-relation $A\subseteq \N_+^r$ can be written as
  \begin{equation} \label{initial-pi^1_n-formula}
  A = \{ \overline{x} \in \N_+^r \mid \forall f_1  \ \exists f_2 \cdots \forall f_n \ \exists y: P(\overline{x} ,y,f_1,\ldots,f_n)\},
  \end{equation}
  where $P$ is a recursive predicate relative to the functions $f_1,
  \ldots, f_n$ (see \cite[p.378]{Odi89}). In other words, there exists
  an oracle Turing-machine which computes the Boolean value
  $P(\overline{x},y,f_1,\ldots,f_n)$ from input
  $(\overline{x},y)$.  The oracle Turing-machine can compute a value
  $f_i(a)$ for a previously computed number $a \in \N_+$ in a single
  step. Therefore we can easily obtain an oracle Turing machine $M$
  which halts on input $\overline{x}$ if and only if $\exists y:
  P(\overline{x},y,f_1,\ldots,f_n)$ holds.

  Following \cite{Odi89}, we can replace the function quantifiers in
  \eqref{initial-pi^1_n-formula} by set quantifiers as follows.  A
  function $f:\N_+\to\N_+$ is encoded by the set $\{ C_2(x,y) \mid
  f(x)=y\}$.  Let $\mathsf{func}(X)$ be the following formula,
  where $X$ is a set variable:
  \begin{eqnarray*}
  \func(X) & = & (\forall x,y,z,u,v  : C_2(x,y)=u \wedge C_2(x,z)=v \wedge u,v \in X \rightarrow y = z) \wedge \\
          &   & (\forall x \ \exists y, z : C_2(x,y) = z \wedge z \in X)
  \end{eqnarray*}
  Hence, $\func(X)$ is a $\Pi_2(\SC,\REC)$-formula, which expresses that
  $X$ encodes a total function on $\N$.  Then, the set $A$ in
  \eqref{initial-pi^1_n-formula} can be defined by the formula
  \begin{equation} \label{equiv-formula-1}
  \forall X_1 : \neg\func(X_1) \vee \exists X_2 : \func(X_2) \wedge \cdots \forall
               X_n : \neg\func(X_n) \vee R(\overline{x},X_1,\ldots,X_n) .
  \end{equation}
  The predicate $R$ can be derived from the oracle Turing-machine $M$
  as follows: Construct from $M$ a new oracle Turing machine $N$ with
  oracle sets $X_1, \ldots, X_n$. If the machine $M$ wants to compute
  the value $f_i(a)$, then the machine $N$ starts to enumerate all $b
  \in \N_+$ until it finds $b \in \N_+$ with $C_2(a,b) \in X_i$.  Then
  it continues its computation with $b$ for $f_i(a)$.  Then the
  predicate $R(\overline{x},X_1,\ldots,X_n)$ expresses that machine
  $N$ halts on input $\overline{x}$.

  Fix a computable bijection $D : \N_+ \to \mathsf{Fin}(\N_+)$, where
  $\mathsf{Fin}(\N_+)$ is the set of all finite subsets of $\N_+$.
  Let $\mathsf{in}(x,y)$ be an abbreviation for $x \in D(y)$. This is
  a computable predicate.

  Next, consider the predicate $R(\overline{x},X_1,\ldots,X_n)$.  In
  every run of the machine $N$ on input $\overline{x}$, the machine
  $N$ makes only finitely many oracle queries.  Hence, the predicate
  $R(\overline{x},X_1,\ldots,X_n)$ is equivalent to
  $$
  \exists b \ \exists (s_1, \ldots, s_n) : S(\overline{x},b,(s_1,\ldots, s_n)) \wedge
      \bigwedge_{i=1}^n \forall z\le b \; (\mathsf{in}(z,s_i)\leftrightarrow z \in X_i),
  $$
  where the predicate $S$ is derived from the Turing-machine $N$ as
  follows: Let $T$ be the Turing-machine that on input
  $(\overline{x},b,(s_1,\ldots, s_n))$ behaves as $N$, but if $N$ asks
  the oracle whether $z \in X_i$, then $T$ first checks whether $z
  \leq b$ (if not, then $T$ diverges) and then checks, whether
  $\mathsf{in}(z,s_i)$ holds. Then $S(\overline{x},b,(s_1,\ldots,
  s_n))$ if and only if $T$ halts on input
  $(\overline{x},b,(s_1,\ldots,s_n))$.  Hence, the predicate
  $S(\overline{x},b,(s_1,\ldots,s_n))$ is recursively enumerable,
  i.e., can be described by a formula from $\Sigma_1(\REC,\SC)$. Hence
  the predicate $R$ can be described by a formula from
  $\Sigma_2(\REC,\SC)$.

  Note that the formula from \eqref{equiv-formula-1} is equivalent
  with a formula 
  \begin{equation}
    \label{eq:DK}
    \forall X_1\exists X_2\cdots\forall X_n:\varphi(\overline x,\overline X),
  \end{equation}
  where $\varphi$ is a Boolean combination of $R$ and formulas of the
  form $\func(X_i)$. Since all these formulas belong to
  $\Pi_2(\REC,\SC)\cup\Sigma_2(\REC,\SC)$, the formula $\varphi$ belongs
  to $B\Sigma_2(\REC,\SC)\subseteq\Pi_3(\REC,\SC)$. Hence \eqref{eq:DK} is
  equivalent with
  \begin{equation}
    \label{eq:DK1}
    \forall X_1\ \exists X_2\cdots\forall X_n
      \forall \overline a \ \exists \overline b \ \forall \overline c : \beta
  \end{equation}
  where $\beta$ is a boolean combination of recursive predicates and
  set constraints.

  We can eliminate the quantifier block $\forall \overline{a}$ by
  merging it with $\forall X_n$: First, we can reduce $\forall
  \overline{a}$ to a single quantifier $\forall a$. For this, assume
  that the length of the tuple $\overline a$ is $k$. Then, $\forall
  \overline{a} \cdots$ in \eqref{eq:DK1} can be replaced by $\forall a
  \ \exists \overline{a} : C_k(\overline{a})=a \wedge \cdots$. Since
  $C_k(\overline{a})=a$ is again recursive and since we can merge
  $\exists \overline a\exists \ \overline b$ into a single block of
  quantifiers $\exists\overline b$, we obtain indeed an equivalent
  formula of the form
  \begin{equation} \label{equiv-formula-3}
  \forall X_1 \ \exists X_2 \cdots \forall  X_n \ \forall a \
   \exists \overline{b} \ \forall \overline{c} : \beta'
  \end{equation}
  where $\beta'$ is a boolean combination of recursive predicates and
  set constraints.

  Next, we encode the pair $(X_n,a)$ by the set $\{
  2x \mid x \in X_n\} \cup \{2a+1\}$.  Let $\alpha(X)$ be the formula
  \begin{eqnarray*}
  \alpha(X) &=& (\forall x,y,x',y' : x = 2x'+1 \wedge y = 2y'+1 \wedge x,y \in X
              \ \to \  x=y) \wedge \\
            & & (\exists x,u : x \in X \wedge x = 2u+1)
  \end{eqnarray*}
  Hence, $\alpha(X)$ expresses that $X$ contains exactly one odd
  number.  Hence, we obtain a formula equivalent to
  \eqref{equiv-formula-3} by
  \begin{itemize}
  \item replacing $\forall X_n \ \forall a \cdots$ with $\forall X_n :
    \neg\alpha(X_n) \vee \exists a, a', a'' : a'' \in X_n \wedge a'' =
    a'+1 \wedge a'=2a \wedge \cdots$ and
  \item replacing every existential quantifier $\exists b_i \cdots$
    (resp.\ universal quantifier $\forall c_i \cdots$) in
    \eqref{equiv-formula-3} with $\exists b_i\ \exists b'_i : b'_i =
    2b_i \wedge \cdots$ (resp.\ $\forall c_i\ \forall c'_i : c'_i \neq
    2c_i \vee \cdots$), and
  \item replacing every sub-formula $a\in X_n$, $b_i \in X_n$ or $c_i
    \in X_n$ with $a'\in X_n$, $b'_i \in X_n$, and $c'_i \in X_n$,
    resp..
  \end{itemize}
  All new quantifiers can be merged with either the block $\exists
  \overline{b}$ or the block $\forall \overline{c}$ in
  \eqref{equiv-formula-3}.  We now have obtained an equivalent formula
  of the form
  \begin{equation} \label{equiv-formula-4}
  \forall X_1 \ \exists X_2 \cdots \forall  X_n \
   \exists \overline{b} \ \forall \overline{c} : \beta''
  \end{equation}
  where $\beta''$ is a Boolean combination of recursive predicates and
  set constraints.

  The block $\exists \overline{b} \cdots$ can be replaced by $\exists
  b\ \forall \overline{b} : C_\ell(\overline{b}) \neq b \vee \cdots$,
  where $\ell$ is the length of the tuple $\overline{b}$. Since
  $C_\ell(\overline{b}) \neq b$ is a computable predicate, this
  results in an equivalent formula of the form
  $$
  \forall X_1 \ \exists X_2 \cdots \forall  X_n \ \exists b \ \forall \overline{c} : \beta'''
  $$
  where $\beta'''$ is a Boolean combination of recursive predicates
  and set constraints.

  Note that the set of recursive predicates is closed under Boolean
  combinations and that the set of set constraints is closed under
  negation. This allows to obtain an equivalent formula of the form
  \[
      \forall X_1 \ \exists X_2 \cdots \forall X_n \ \exists b \ \forall \overline{c} : \bigwedge_{i=1}^\ell (R_i \lor \psi_i),
  \]
  where the $R_i$ are recursive predicates and the $\psi_i$ are disjunctions
  of set constraints. 

  Since the recursive predicates $R_i$ are co-Diophantine, there are
  polynomials $p_i,q_i\in\N[b,\overline c,\overline z]$ such that
  $R_i(b,\overline c)$ is equivalent with $\forall \overline
  z:p_i(b,\overline c,\overline z)\neq q_i(b,\overline c,\overline
  z)$. Replacing $R_i$ in the above formula by this equivalent formula
  and merging the new universal quantifiers $\forall \overline z$ with
  $\forall\overline c$ results in a formula as required. \qed
\end{proof}
It is known that the first-order quantifier block $\exists y \, \forall
\overline{z}$ in Proposition~\ref{prp:normal.form} cannot be replaced by a
block with only one type of first-order quantifiers, see
e.g.\ \cite{Odi89}.

\section{$\omega$-automatic trees of height at least
  $4$} \label{sec:h4}

We prove the following theorem for injectively $\omega$-automatic
trees of height at least $4$.

\begin{theorem}\label{thm:h4}
  Let $n\ge1$ and $\Theta\in\{\Sigma,\Pi\}$. There exists a tree
  $U_{n,\Theta}$ of height $n+3$ such that the set $\{P\in\Ti_{n+3}\mid \S(P)\cong
  U_{n,\Theta}\}$ is hard for $\Theta^1_n$.  Hence,
  \begin{itemize}
  \item  the isomorphism
  problem $\Iso(\Ti_{n+3})$ for the class of injectively
  $\omega$-automatic trees of height~$n+3$ is hard for both the
  classes $\Pi^1_n$ and $\Sigma^1_n$,
  \item and  the isomorphism
  problem $\Iso(\Ti)$ for the class of injectively $\omega$-automatic
  trees of finite height is not analytical.
  \end{itemize}
  \end{theorem}
Theorem~\ref{thm:h4} will be derived from the following proposition whose proof occupies
Sections~\ref{sec:h4.construction} and~\ref{sec:h4.automata}.

\begin{proposition}\label{prop:h4}
  Let $n\ge1$. There are trees $U[0]$ and $U[1]$ of height $n+3$ such
  that for any set $A\subseteq\N_+$ that is $\Pi^1_{n}$ if $n$ is odd and
  $\Sigma^1_{n}$ if $n$ is even, one can compute from $x\in\N_+$ an
  injectively $\omega$-automatic tree $T[x]$ of height $n+3$ with $T[x]\cong U[0]$ if
  and only if $x\in A$ and $T[x]\cong U[1]$ otherwise.
\end{proposition}

\noindent
{\em  Proof of Theorem~\ref{thm:h4} from Proposition~\ref{prop:h4}.}
  Let $n\ge1$ be odd. Let $A$ be an arbitrary set from $\Pi^1_n$ and
  set $U_{n,\Pi}=U[0]$ and $U_{n,\Sigma}=U[1]$. Then the mapping
  $x\mapsto T[x]$ is a reduction from $A$ to $\{P\in\Ti_{n+3}\mid
  \S(P)\cong U_{n,\Pi}\}$ and, at the same time, a reduction from the
  $\Sigma^1_n$-set $\N_+\setminus A$ to $\{P\in\Ti_{n+3}\mid
  \S(P)\cong U_{n,\Sigma}\}$. Since $A$ was chosen arbitrary from
  $\Pi^1_n$, the first statement follows for $n$ odd. If $n$ is even,
  we can proceed similarly exchanging the roles of $U[0]$ and $U[1]$.

  We now derive the second statement. By the first one, the trees
  $U[0]$ and $U[1]$ are in particular injectively $\omega$-automatic
  and of height $n+3$, so let $P_0$ and $P_1$ be injective
  $\omega$-automatic presentations of these two trees. Then
  $P\mapsto(P,P_0)$ is a reduction from the set $\{P\in\Ti_{n+3}\mid
  \S(P)\cong U_{n,\Pi}\}$ to $\Iso(\Ti_{n+3})$ which is therefore hard
  for $\Pi^1_{n+3}$. Analogously, this isomorphism problem is hard for
  $\Sigma^1_{n+3}$.

  Finally, we prove the third statement. For any $n\ge1$, the set
  $\Ti_{n+3}$ is decidable (since the set of trees of height at most 3
  is first-order axiomatizable). With $P',P''\in\Ti_{n+3}$ arbitrary
  with $\S(P')\not\cong\S(P'')$, the mapping
  \[
    (P_1,P_2)\mapsto
    \begin{cases}
      (P_1,P_2) & \text{ if }P_1,P_2\in\Ti_{n+3}\\
      (P',P'') & \text{ otherwise}
    \end{cases}
  \]
  is a reduction from $\Iso(\Ti_{n+3})$ to $\Iso(\Ti)$. Hence
  $\Iso(\Ti)$ is hard for all levels $\Sigma^1_n$ and therefore not
  analytical.\qed

\medskip
\noindent
The construction of the trees $T[x]$, $U[0]$, and $U[1]$ is uniform in
$n$ and the formula defining $A$. Hence the second-order theory of
$(\N,+,\times)$ can be reduced to
$\bigcup_{n\geq 1}\{n\}\times\Iso(\Ti_n)$ and therefore to the
isomorphism problem $\Iso(\bigcup_{n\geq 1} \Ti_n)$. 

\begin{corollary}\label{C-2nd-order-arithmetic}
  The second-order theory of $(\N,+,\times)$ can be reduced to the
  isomorphism problem $\Iso(\bigcup_{n\in \N_+} \Ti_n)$ for the class
  of all injectively $\omega$-automatic trees of finite height.
\end{corollary}
We now start to prove Proposition~\ref{prop:h4}. Let $A$ be a set that is
$\Pi^1_{n}$ if $n$ is odd and $\Sigma^1_{n}$ otherwise. By
Proposition~\ref{prp:normal.form} it can be written in the form
\begin{equation*}
    A = \{x \in \N_+\mid Q_1 X_1 \ldots Q_{n} X_{n} \exists y \; \forall \overline{z}: \bigwedge^{\ell}_{i=1} p_i(x,y,\overline{z})\neq q_i(x,y,\overline{z}) \vee \psi_i(x,y,\overline{z},\overline X)\}
  \end{equation*}
where
\begin{itemize}
\item $Q_1, Q_2,\ldots,Q_n$ are alternating quantifiers with
  $Q_{n}=\forall$,
\item $p_i,q_i$ $(1\leq i\leq \ell)$ are polynomials in
  $\N[x,y,\overline{z}]$ where $\overline z$ has length $k$, and
\item every $\psi_i$ is a disjunction of set constraints on the set
  variables $X_1,\dots,X_{n}$ and the individual variables $x, y,
  \overline{z}$.
\end{itemize}
Let $\varphi_{-1}(x,y,X_1,\dots,X_n)$ be the formula
\begin{equation*}
  \forall \overline z:\bigwedge^{\ell}_{i=1}
  p_i(x,y,\overline{z})\neq q_i(x,y,\overline{z})
  \vee \psi_i(x,y,\overline{z},\overline X)\,.
\end{equation*}
For $0\le m\le n$, we will also consider the formula
$\varphi_m(x,X_1,\dots,X_{n-m})$ defined by
\begin{equation*}
Q_{n+1-m} X_{n+1-m}\ldots Q_{n} X_{n}\;\exists
y:\varphi_{-1}(x,y,X_1,\dots,X_n)
\end{equation*}
such that $\varphi_0(x,X_1,\dots,X_n)$ is a first-order formula and
$\varphi_n(x)$ holds if and only if $x\in A$. 

To prove Proposition~\ref{prop:h4}, we construct by induction on $0\le m \le
n$ height-$(m+3)$ trees $T_m[X_1,\ldots,X_{n-m}, x]$ and $U_m[i]$
where $X_1,\ldots,X_{n-m} \subseteq \N_+$, $x\in \N_+$, and $i\in
\{0,1\}$ such that the following holds:
\begin{equation}
  \forall \overline X\in (2^{\N_+})^{n-m} \; \forall x \in
  \N_+:
  T_m[\overline X,x] \cong
    \begin{cases}
      U_m[0] & \text{ if }\varphi_m(x,\overline X) \text{ holds}\\
      U_m[1] & \text{ otherwise }
    \end{cases}\label{eq:invariant}
\end{equation}
Setting $T[x]=T_n[x]$, $U[0]=U_n[0]$, and $U[1]=U_n[1]$ and
constructing from $x$ an injectively $\omega$-automatic presentation
of~$T[x]$ then proves Proposition~\ref{prop:h4}.

\subsection{Construction of trees}
\label{sec:h4.construction}

In the following, we will use the 
{\em injective} polynomial function
\begin{equation}
  \label{eq:C}
  C:\N_+^2\to\N_+ \text{ with } C(x,y)=(x+y)^2+3x+y .
\end{equation}
For $e_1,e_2\in\N_+$, let $S[e_1,e_2]$ denote the height-1 tree
containing $C(e_1,e_2)$ leaves.  For $(\overline
X,x,y,\overline{z},z_{k+1})\in (2^{\N_+})^n\times \N_+^{k+3}$ and
$1\leq i\leq \ell$, define the following height-1 tree, where
$\ell$, $p_i$, and $q_i$ refer to the definition of the set $A$ above:\footnote{The
  choice of $S[1,2]$ in the first case is arbitrary. Any $S[a,b]$ with
  $a\neq b$ would be acceptable.}
\begin{equation}\label{T[X,x,y,z,i]}
    T'[\overline{X},x,y,\overline{z},z_{k+1},i] =
     \begin{cases}
       S[1,2] & \text{if }\psi_i(x,y,\overline{z}, \overline{X})\\
       S[p_i(x,y,\overline{z})+z_{k+1}, q_i(x,y,\overline{z})+z_{k+1}]
          & \text{otherwise.}
    \end{cases}
\end{equation}
Next, we define the following height-2 trees, where $\kappa \in \N_+\cup\{\omega\}$ (we consider
the natural order on $\N_+\cup\{\omega\}$ with $n < \omega$ for all $n \in \N_+$):
\begin{align}
\label{eq:T[X,x,y]}
T''[\overline X,x,y] &= r \circ
            \left(
              \begin{array}{l}\displaystyle
               \biguplus \{S[e_1,e_2] \mid e_1\neq e_2\} \ \uplus  \\
               \displaystyle
               \biguplus \{ T'[\overline X,x,y,\overline{z},z_{k+1},i]
                     \mid \overline{z}\in \N^k_+, z_{k+1}\in \N_+, 1\leq i\leq \ell\}
             \end{array}
            \right)^{\aleph_0}\\
\label{eq:U_kappa}
U''[\kappa] &= r \circ
            \left(
               \biguplus \{S[e_1,e_2] \mid e_1\neq e_2\}
               \uplus
               \biguplus \{S[e,e] \mid \kappa\le e < \omega\}
            \right)^{\aleph_0} .
\end{align}
Note that all the trees $T''[\overline X,x,y]$ and $U''[\kappa]$ are
build from trees of the form $S[e_1,e_2]$. Furthermore, if $S[e,e]$
appears as a building block, then $S[e+a,e+a]$ also appears as one for
all $a\in\N$. In addition, any building block $S[e_1,e_2]$ appears
either infinitely often or not at all. 
In this sense, $U''[\kappa]$ encodes the set of pairs
$\{(e_1,e_2)\mid e_1\neq e_2\}\cup\{(e,e)\mid \kappa\le e<\omega\}$
and $T''[\overline X,x,y]$ encodes  the set of pairs $\{(e_1,e_2)\mid
e_1\neq e_2\}\cup$ $\{(p_i(x,y,\overline z)+z_{k+1},q_i(x,y,\overline
z)+z_{k+1})\mid 1\le i\le \ell,x,y,z_{k+1}\in\N_+,\overline z\in
\N_+^k\}$. These observations allow to prove the following:

\begin{lemma}\label{lem:h3.T[X,x,y]}
  Let $\overline X\in(2^{\N_+})^n$ and $x,y\in\N_+$. Then the
  following hold:
  \begin{enumerate}[(a)]
  \item $T''[\overline X,x,y] \cong U''[\kappa]$ for some $\kappa\in
    \N_+\cup \{\omega\}$
  \item $T''[\overline X,x,y] \cong U''[\omega]$ if and only if
    $\varphi_{-1}(x,y,\overline X)$ holds\label{(b)}
  \end{enumerate}
\end{lemma}

\begin{proof}
  Let us start with the second property.  Suppose
  $\varphi_{-1}(x,y,\overline X)$ holds. Let $\overline{z}\in
  \N^{k}_+$, $z_{k+1} \in \N$, and $1\leq i\leq \ell$. Since
  $p_i(x,y,\overline z)\neq q_i(x,y,\overline z)$, there are natural
  numbers $e_1\neq e_2$ with 
  $T'[\overline X,x,y,\overline z, z_{k+1},i]=S[e_1,e_2]$. 
  Hence $T''[\overline X,x,y]\cong U''[\omega]$.

  Conversely, suppose $T''[\overline X,x,y]\cong U_\omega$. Let $\overline
  z\in\N^k$, $z_{k+1}\in\N$, and $1\le i\le\ell$. Then
  $T'[\overline X,x,y,\overline z,z_{k+1},i]$ is a height-2 subtree of
  $T''[\overline X,x,y]\cong U''[\omega]$. Hence there are natural numbers
  $e_1\neq e_2$ with $T'[\overline X,x,y,\overline z,z_{k+1},i]\cong
  S[e_1,e_2]$. By \eqref{T[X,x,y,z,i]}, this implies
  $p_i(x,y,\overline z)\neq q_i(x,y,\overline z) \lor \psi_i(x,y, \overline z, \overline X)$. Hence we showed that $\forall \overline{z}:
  \bigwedge_{i=1}^\ell p_i(x,y,\overline{z}) \neq
  q_i(x,y,\overline{z}) \vee \psi_i(x,y,\overline{z},\overline X)$ holds.

  Now it suffices to prove the first statement in case
  $\varphi_{-1}(x,y,\overline X)$ does not hold. Then there exist some
  $\overline z\in\N_+^k$ and $1\le i \le \ell$ with 
  \[
    p_i(x,y,\overline z)=q_i(x,y,\overline z)\land
     \neg\psi_i(x,y,\overline z,\overline X)\,.
  \]
  Hence there is some $e\in\N_+$ such that $S[e,e]$ appears in the
  definition of $T''[\overline X,x,y]$. Let $m=\min\{e\in\N_+\mid
  S[e,e]\text{ appears in }T''[\overline X,x,y]\}$. Then, for all
  $a\in\N$, also $S[m+a,m+a]$ appears in $T''[\overline X,x,y]$. 
  Hence $T''[\overline X,x,y]\cong U''[m]$. \qed
\end{proof}
In a next step, we collect the trees $T''[\overline X,x,y]$ and
$U''[\kappa]$ into the trees $T_0[\overline X,x]$, $U_0[0]$, and
$U_0[1]$ as follows:
\begin{align*}
  T_0[\overline X,x] &=
    r\circ
     \left(
      \biguplus\{U''[m]\mid m\in \N_+\}
      \uplus
      \biguplus\{T''[\overline X,x,y]\mid y\in \N_+\}
     \right)^{\aleph_0}\\
  U_0[0] &=
    r\circ
     \left(
      \biguplus\{U''[\kappa]\mid \kappa\in \N_+ \cup \{\omega\}\}
     \right)^{\aleph_0}\\
  U_0[1] &=
    r\circ
     \left(
      \biguplus\{U''[m]\mid m\in \N_+\}
     \right)^{\aleph_0}
\end{align*}
By Lemma~\ref{lem:h3.T[X,x,y]}(a), these trees are build from copies
of the trees $U''[\kappa]$ (and are therefore of height~3), each
appearing either infinitely often or not at all. 

\begin{lemma}\label{lem:h4+.T_3}
  Let $\overline X \in (2^{\N_+})^n$ and $x\in \N_+$. Then
  \[
      T_0[\overline X, x] \cong 
      \begin{cases}
        U_{0}[0] & \text{ if }\varphi_{0}(x,\overline X)
                        \text{ holds and}\\
        U_{0}[1] & \text{ otherwise.}
      \end{cases}
  \]
\end{lemma}

\begin{proof}
  If $T_0[\overline X,x]\cong U_0[0]$, then there must be some $y\in
  \N_+$ such that $T''[\overline X,x,y]\cong U''[\omega]$. By
  Lemma~\ref{lem:h3.T[X,x,y]}(b), this means that
  $\varphi_0(x,\overline{X})$ holds.

  On the other hand, suppose $T_0[\overline X,x]\not\cong
  U_0[0]$. Then $T''[\overline X,x,y]\not\cong U''[\omega]$ for all
  $y\in\N_+$. From Lemma~\ref{lem:h3.T[X,x,y]}(b) again, we obtain
  for all $y\in \N_+$:
  $T''[\overline X,x,y]\cong U''[m_y]$ for some $m_y\in\N_+$. 
  Hence $T_0[\overline X,x]\cong U_0[1]$ in this case.
  \qed
\end{proof}
Now, we come to the induction step in the construction of our trees.
Suppose that for some $0\leq m <n$ we have height-$(m+3)$ trees
$T_m[X_1,\dots,X_{n-m},x]$, $U_m[0]$ and $U_m[1]$ satisfying
\eqref{eq:invariant}. Let $\overline X$ stand for
$(X_1,\dots,X_{n-m-1})$ and let $\alpha=m\bmod 2$. We define the
following height-$(m+4)$ trees:
\begin{eqnarray*}
T_{m+1}[\overline X,x] &=&
\displaystyle
  r\circ \left( U_{m}[\alpha]
    \uplus\biguplus\big\{T_{m}[\overline X,X_{n-m},x] \mid
    X_{n-m}\subseteq \N_+\big\}\right)^{2^{\aleph_0}}\\
U_{m+1}[i] &=&
  r\circ \left( U_{m}[\alpha] \uplus U_{m}[i] \right)^{2^{\aleph_0}}
  \text{ for }i \in \{0,1\}
\end{eqnarray*}
Note that the trees $T_{m+1}[\overline X,x]$, $U_{m+1}[0]$, and
$U_{m+1}[1]$ consist of $2^{\aleph_0}$ many copies of $U_m[\alpha]$
and possibly $2^{\aleph_0}$ many copies of $U_m[1-\alpha]$. 

\begin{lemma} \label{lem:h4+.T_m} Let $X_1,\ldots,X_{n-m-1}\subseteq
  \N_+$ and $x\in \N_+$. Then
  \[
      T_{m+1}[X_1,\ldots,X_{n-m-1},x] \cong 
      \begin{cases}
        U_{m+1}[0] & \text{ if }\varphi_{m+1}(x,X_1,\ldots X_{n-m-1})
                        \text{ holds}\\
        U_{m+1}[1] & \text{ otherwise.}
      \end{cases}
  \]
\end{lemma}

\begin{proof}
  We have to handle the cases of odd and even $m$ separately and start
  assuming $m$ to be even (i.e., $\alpha=0$) such that the outermost
  quantifier $Q_{n-m}$ of the formula $\varphi_{m+1}(x,X_1,\dots,X_{n-m-1})$ is
  universal.

  Suppose that $\varphi_{m+1}(X_1,\dots,X_{n-m-1},x)$ holds. Then, by
  the inductive hypothesis, for each $X_{n-m}\subseteq\N_+$,
  $T_m[X_1,\ldots,X_{n-m},x]\cong U_m[0]$.  Hence all height-$(m+3)$
  subtrees of $T_{m+1}[X_1,\ldots,X_{n-m-1},x]$ are isomorphic to
  $U_m[0]$ and thus
  \[
     T_{m+1}[X_1,\ldots,X_{n-m-1},x]\cong 
         r\circ U_m[0]^{2^{\aleph_0}} = U_{m+1}[0]\,.
  \]
  On the other hand, suppose that
  $\neg\varphi_{m+1}(X_1,\dots,X_{n-m-1},x)$ holds. Then there exists
  some set $X_{n-m}$ such that $\neg\varphi_m(X_1,\dots,X_{n-m},x)$ is
  true. Hence, by the induction hypothesis, 
  \[
   T_m(X_1,\dots,X_{n-m},x)\cong U_m[1],
  \]
  i.e., $T_{m+1}(X_1,\dots,X_{n-m-1},x)$ contains one (and therefore
  $2^{\aleph_0}$ many) height-$(m+3)$ subtrees isomorphic to
  $U_m[1]$. This implies $T_{m+1}(X_1,\dots,X_{n-m-1},x)\cong
  U_{m+1}[1]$ since $m$ is even.

  The arguments for $m$ odd are very similar and therefore left to the
  reader.  \qed
\end{proof}
The following lemma follows from Lemma~\ref{lem:h4+.T_m} with $m=n$
and the fact that $\varphi_n(x)$ holds if and only if $x\in A$.

\begin{lemma}\label{lem:h4+.reduction}
  For all $x \in \N_+$, we have $T_n[x] \cong U_n[0]$ if $x\in A$ and
  $T_n[x]\cong U_n[1]$ otherwise.
\end{lemma}

\subsection{Injective $\omega$-automaticity}
\label{sec:h4.automata}

Injectively $\omega$-automatic presentations of the trees $T_m[\overline
X,x]$, $U_m[0]$, and $U_m[1]$ will be constructed inductively. Note
that the construction of $T_{m+1}[\overline X,x]$ involves all the
trees $T_m[\overline X,X_{n-m},x]$ for $X_{n-m}\subseteq\N_+$. Hence
we need \emph{one single injectively $\omega$-automatic presentation}
for the forest consisting of all these trees. Therefore, we will deal
with forests. To move from one forest to the next, we will always
proceed as follows: add a set of new roots and connect them to some of
the old roots \emph{which results in a directed acyclic graph} (or
dag) and not necessarily in a forest. The next forest will then be the
unfolding of this dag.

The {\em height} of a dag $D$ is the length (number of edges) of a
longest directed path in~$D$. We only consider dags of finite
height. A {\em root} of a dag is a node without incoming edges. A dag
$D=(V,E)$ can be unfolded into a forest $\unfold(D)$ in the usual way:
Nodes of $\unfold(D)$ are directed paths in $D$ that start in a root
and the order relation is the prefix relation between these paths. For
a root $v\in V$ of $D$, we define the tree $\unfold(D,v)$ as the
restriction of $\unfold(D)$ to those paths that start in $v$. We will
make use of the following lemma whose proof is based on the immediate
observation that the set of convolutions of paths in $D$ is again a
regular $\omega$-language.

\begin{lemma}\label{lem:dag}
  From a given $k\in \N$ and an injectively $\omega$-automatic
  presentation for a dag $D$ of height at most $k$, one can construct
  effectively an injectively $\omega$-automatic presentation for
  $\unfold(D)$ such that the roots of $\unfold(D)$ coincide with the
  roots of~$D$ and $\unfold(D,r)=(\unfold(D))(r)$ for any root~$r$.
\end{lemma}

\begin{proof}
  Let $D =(V,E) = \S(P)$, i.e., $V$ is an $\omega$-regular language and the
  binary relation $E\subseteq V\times V$ is $\omega$-automatic. The
  universe for our injectively $\omega$-automatic copy of $\unfold(D)$
  is the set $L$ of all convolutions $v_0\otimes v_1\otimes v_2\otimes
  \cdots \otimes v_m$, where $v_0$ is a root and $(v_i,v_{i+1})\in E$
  for all $0\leq i<m$. Since the dag $D$ has height at most $k$, we have
  $m\leq k$. Since the edge relation of $D$ is $\omega$-automatic and
  since the set of all roots in $D$ is $\FO$-definable and hence
  $\omega$-regular by Theorem\ref{thm:extFOaut}, $L$ is indeed an
  $\omega$-regular set. Moreover, the edge relation of $\unfold(D)$
  becomes clearly $\omega$-automatic on~$L$. \qed
\end{proof}
For a symbol $a$ and a tuple $\overline{e}=(e_1,\ldots,e_k)\in
\N^k_+$, we write $a^{\overline{e}}$ for the $\omega$-word
\[
    a^{e_1} \otimes a^{e_2} \otimes \cdots \otimes a^{e_k}=
    (a^{e_1}\diamond^\omega) \otimes (a^{e_2}\diamond^\omega) \otimes \cdots \otimes (a^{e_k}\diamond^\omega)\,.
\]
For an $\omega$-language $L$, we write $\otimes_k(L)$ for
$\otimes(L^k)$. The following lemma was shown in \cite{KuLiLo10}
for finite words instead of $\omega$-words.

\begin{lemma}\label{lem:polynomial.automaton}
  Given a non-zero polynomial $p(\overline{x}) \in \N[\overline{x}]$
  in $k$ variables, one can effectively construct a B\"uchi automaton
  $\B[p(\overline{x})]$ over the alphabet $\{a,\diamond\}^k$ with
  $L(\B[p(\overline{x})]) = \otimes_k(a^+)$ such that for all
  $\overline{c}\in \N^k_+: \B[p(\overline{x})]$ has exactly
  $p(\overline{c})$ accepting runs on input $a^{\overline{c}}$.
\end{lemma}

\begin{proof}
  B\"uchi automata for the polynomials $p(\overline x)=1$ and
  $p(\overline x)=x_i$ are easily build. Inductively, let
  $\B[p_1(\overline x)+p_2(\overline x)]$ be the disjoint union of
  $\B[p_1(\overline{x})]$ and $\B[p_2(\overline{x})]$ and let
  $\B[p_1(\overline x)\cdot p_2(\overline x)]$ be obtained from
  $\B[p_1(\overline{x})]$ and $\B[p_2(\overline{x})]$ by the flag
  construction.\qed
\end{proof}
For $X\subseteq\N_+$, let $w_X\in\{0,1\}^*$ be the
characteristic word (i.e., $w_X[i]=1$ if and only if $i \in X$) and,
for $\overline X = (X_1,\ldots,X_n)\in(2^{\N_+})^n$, write
$w_{\overline X}$ for the convolution of the words~$w_{X_i}$.

\begin{lemma}\label{lem:h3.A_psi}
  From a given Boolean combination
  $\psi(x_1,\ldots,x_m,X_1,\ldots,X_n)$ of set constraints on set
  variables $X_1,\ldots,X_n$ and individual variables $x_1,\ldots,x_m$
  one can construct effectively a deterministic B\"uchi automaton
  $\A_{\psi}$ over the alphabet $\{0,1\}^n\times\{a,\diamond\}^m$ such
  that for all $X_1,\ldots,X_n\subseteq \N_+, \overline{c}\in \N^m_+$,
  the following holds:
  \[
      w_{X_1}\otimes\cdots \otimes w_{X_n} \otimes a^{\overline{c}} \in L(\A_{\psi}) 
      \ \iff \ \psi(\overline{c}, X_1,\ldots,X_n) \text{ holds.}
  \]
\end{lemma}

\begin{proof}
  We can assume that $\psi$ is a positive Boolean combination,
  since the $\omega$-word $w_{\N_+\setminus X}$ is simply obtained
  from $w_X$ by exchanging the symbols $0$ and $1$. Then
  the claim is trivial for a single set constraint. Since
  $\omega$-languages accepted by deterministic B\"uchi automata are
  effectively closed under intersection and union, the result follows.\qed
\end{proof}

\begin{lemma}\label{lem:h3.A[varphi]}
  For $1\le i\le\ell$, there exists a B\"uchi-automaton $\A_i$ with
  the following property: For all $\overline X\in (2^{\N_+})^n$,
  $\overline z\in\N_+^k$, and $x,y,z_{k+1}\in\N_+$, the number of
  accepting runs of $\A_i$ on the word $w_{\overline X} \otimes
  a^{(x,y,\overline z,z_{k+1})}$ equals 
  $$
  \begin{cases}
    C(1,2) & \text{ if }   \psi_i(x,y,\overline z, \overline X) \text{ holds}\\
    C(p_i(x,y,\overline z) + z_{k+1}, q_i(x,y,\overline z)+z_{k+1}) &\text{
      otherwise.}
  \end{cases}
  $$
  \end{lemma}

\begin{proof}
  By Lemma~\ref{lem:polynomial.automaton}, one can construct a B\"uchi
  automaton $\B_i$, which has precisely $C(p_i(x,y,\overline z) + z_{k+1},
  q_i(x,y,\overline z)+z_{k+1})$ many accepting runs on the $\omega$-word $w_{\overline X}\otimes a^{(x,y,\overline
    z,z_{k+1})}$. Secondly, one
  builds deterministic B\"uchi automata $\C_i$ and $\overline\C_i$
  accepting a word $w_{\overline X}\otimes a^{(x,y,\overline
    z,z_{k+1})}$ if and only if the disjunction $\psi_i(x,y,\overline
  z,\overline X)$ of set constraints is satisfied (not satisfied,
  resp.) which is possible by Lemma~\ref{lem:h3.A_psi}.

  Let $\A$ be the result of applying the flag construction to
  $\overline\C_i$ and $\B_i$.  If  $\overline X\in (2^{\N_+})^n$,
  $\overline z\in\N_+^k$, and $x,y,z_{k+1}\in\N_+$,
   then the number of accepting
  runs of $\A$ on the word $w_{\overline X} \otimes
  a^{(x,y,\overline z,z_{k+1})}$
   equals
  \[
  \begin{cases}
    0 & \text{ if }   \psi_i(x,y,\overline z, \overline X) \text{ holds}\\
    C(p_i(x,y,\overline z) + z_{k+1}, q_i(x,y,\overline z)+z_{k+1}) &\text{
      otherwise.}
  \end{cases}
  \]
  Hence the disjoint union of $\A$ and $C(1,2)$ many copies of $\C_i$
  has the desired properties. \qed
\end{proof}

\begin{proposition}\label{prop:H'}
  There exists an injectively $\omega$-automatic forest $\H'=(L',E')$
  of height 1 such that
  \begin{itemize}
  \item the set of roots equals
    $\{1,\ldots,\ell\}\otimes(\{0,1\}^\omega)^n\otimes(\otimes_{k+3}(a^+))\cup
    (b^+\otimes b^+)$,
  \item for $1\le i\le\ell$, $\overline X\in(2^{\N_+})^n$,
    $x,y,z_{k+1}\in\N_+$ and $\overline z\in\N_+^k$, we have
    \[
    \H'(i\otimes w_{\overline X}\otimes a^{(x,y,\overline z,z_{k+1})})\cong
     T'[\overline X,x,y,\overline z,z_{k+1},i]\text{ and }
    \]
  \item for $e_1,e_2\in\N_+$, we have
    \[
      \H'(b^{(e_1,e_2)})\cong S[e_1,e_2]\,.
    \]
  \end{itemize}
\end{proposition}

\begin{proof}
  Using Lemma~\ref{lem:polynomial.automaton} (with the polynomial
  $p=C(x_1,x_2)$) and Lemma~\ref{lem:h3.A[varphi]}, we can construct a
  B\"uchi-automaton $\A$ accepting
  $\{1,\ldots,\ell\}\otimes(\{0,1\})^n\otimes(\otimes_{k+3}(a^+))\cup
  (b^+\otimes b^+)$ such that the number of accepting runs of $\A$ on
  the $\omega$-word $u$ equals
  \begin{enumerate}[(i)]
  \item $C(e_1,e_2)$ if $u=b^{(e_1,e_2)}$,
  \item $C(1,2)$ if $u=i\otimes w_{\overline X}\otimes
    a^{(x,y,\overline z,z_{k+1})}$ such that $\psi_i(x,y,\overline
    z,\overline X)$ holds, and
  \item $C(p_i(x,y,\overline z)+z_{k+1},q_i(x,y,\overline z)+z_{k+1})$
    if $u=i\otimes w_{\overline X}\otimes a^{(x,y,\overline
      z,z_{k+1})}$ such that $\psi_i(x,y,\overline z,\overline X)$
    does not hold.
  \end{enumerate}
  Let $\Run_\A$ denote the set of accepting runs of $\A$. Note that
  this is a regular $\omega$-language over the alphabet $\Delta$ of
  transitions of $\A$. Now the forest $\H'$ is defined as follows:
  \begin{itemize}
  \item Its universe equals $L(\A)\cup\Run_\A$.
  \item There is an edge $(u,v)$ if and only if $v\in\Run_\A$ is a
    accepting run of $\A$ on $u\in L(\A)$.
  \end{itemize}
  It is clear that $\H'$ is an injectively $\omega$-automatic
  forest of height 1 with set of roots $L(\A)$ as required. Note that
  (i)-(iii)  describe the number of leaves of the height-1 tree
  rooted at $u\in L(\A)$. By (i), we therefore get immediately
  $\H'(b^{(e_1,e_2)})\cong S[e_1,e_2]$. Comparing the numbers in (ii)
  and (iii) with the definition of the tree
  $T'[\overline{X},x,y,\overline{z},z_{k+1},i]$ in
  \eqref{T[X,x,y,z,i]} completes the proof.\qed
\end{proof}
{}From $\H'=(L',E')$, we build an injectively $\omega$-automatic dag
$\D$ as follows:
\begin{itemize}
\item The domain of $\D$ is the set
  $(\otimes_n(\{0,1\}^\omega) \otimes a^+ \otimes a^+)
      \ \cup \ b^* \
      \cup \ \big(\$^* \otimes L')$.
\item For $u,v\in L'$, the words $\$^i\otimes u$ and $\$^j\otimes v$
  are connected if and only if $i=j$ and $(u,v)\in E'$. In other
  words, the restriction of $\D$ to $\$^* \otimes L'$ is isomorphic to
  $\H'^{\aleph_0}$.
\item For all $\overline X\in (2^{\N_+})^n$, $x,y\in \N_+$, the new
  root $w_{\overline X}\otimes a^{(x,y)}$ is connected to all nodes in
  \[
    \$^*
    \otimes
     \left(
       (\{1,\ldots,\ell\}
        \otimes w_{\overline X} \otimes a^{(x,y)} \otimes (\otimes_{k+1}(a^+))
       )
       \cup \{b^{(e_1,e_2)}\mid e_1 \neq e_2\}
     \right)\,.
  \]
\item The new root $\varepsilon$ is connected to all nodes in
   $\$^* \otimes \{b^{(e_1,e_2)}\mid e_1\neq e_2\}$.

\item For all $m\in \N_+$, the new root $b^m$ is connected to all nodes
  in
  \[
   \$^* \otimes \{b^{(e_1,e_2)}\mid e_1\neq e_2 \vee e_1=e_2 \geq m\}.
  \]
\end{itemize}
It is easily seen that $\D$ is an injectively $\omega$-automatic
dag. Let $\H''=\unfold(\D)$ which is also injectively
$\omega$-automatic by Lemma~\ref{lem:dag}. Then, for all $\overline X\in(2^{\N_+})^n$,
    $x,y,m\in\N_+$, we have
\begin{eqnarray*}\label{computations-H''}
    \H''(w_{\overline X}\otimes a^{(x,y)}) &\cong&  
       (w_{\overline X}\otimes a^{(x,y)}) \circ
        \left(
        \begin{array}{l}
           \biguplus\{\H'(i\otimes w_{\overline X}\otimes a^{(x,y,\overline{z})}) \mid
                             1\leq i\leq \ell,\overline{z}\in \N^{k+1}_+\} \uplus \\[1.5mm]
             \biguplus\{\H'(b^{(e_1,e_2)})\mid  e_1\neq e_2 \}
           \end{array}
        \right)^{\aleph_0} \\
    & \stackrel{\text{Prop.~\ref{prop:H'}}}{\cong} & r \circ 
        \left(
        \begin{array}{l}
          \biguplus\{ T'[X,x,y,\overline{z},i] \mid 
                        \overline{z}\in \N^{k+1}_+, 1\leq i\leq \ell\} 
          \uplus \\[1.5mm]
          \biguplus \{S[e_1,e_2] \mid  e_1 \neq e_2\}
        \end{array}
        \right)^{\aleph_0} \\
    & \stackrel{\eqref{eq:T[X,x,y]}}{=} & T''[X,x,y]\\[1em]
    \H''(\varepsilon) &{\cong} &  \varepsilon \circ 
      \Big(\biguplus\{\H'(b^{(e_1,e_2)})\mid e_1\neq e_2 \}\Big)^{\aleph_0} \\
    &\stackrel{\text{Prop.~\ref{prop:H'}}}{\cong} & r \circ \biguplus  \Big(\{S[e_1,e_2] \mid  e_1\neq e_2\}\Big)^{\aleph_0} \\
    &\stackrel{\eqref{eq:U_kappa}}{=} & U''[\omega] \\[1em]
    \H''(b^m) &{\cong} &  b^m \circ 
        \Big(
          \biguplus\{\H'(b^{(e_1,e_2)})\mid 
                          e_1 \neq e_2 \vee e_1 = e_2 \geq m\}
        \Big)^{\aleph_0} \\
    &\stackrel{\text{Prop.~\ref{prop:H'}}}{\cong}& r \circ 
        \Big(
          \biguplus\{S[e_1,e_2] \mid 
                          e_1 \neq e_2 \vee e_1=e_2 \geq m\}
        \Big)^{\aleph_0} \\
    &\stackrel{\eqref{eq:U_kappa}}{=} & U''[m]
\end{eqnarray*}
{}From $\H''=(L'',E'')$, we build an injectively $\omega$-automatic dag
$\D_0$ as follows:
\begin{itemize}
\item The domain of $\D_0$ is the set
  $(\otimes_n\{0,1\}^\omega)\otimes a^+
  \cup\{\varepsilon,b\}\cup(\$^*\otimes L'')$.
\item For $u,v\in L''$, the words $\$^i\otimes u$ and $\$^j\otimes v$
  are connected by an edge if and only if $i=j$ and $(u,v)\in E''$,
  i.e., the restriction of $\D_0$ to $\$^*\otimes L''$ is isomorphic
  to ${\H''}^{\aleph_0}$.
\item For $\overline X\in (2^{\N_+})^n$, $x\in \N_+$, connect the new
  root $w_{\overline X} \otimes a^x$ to
  all nodes in
  \[
  \$^* \otimes \left( w_{\overline X} \otimes a^x \otimes a^+ \;\cup\; b^+ \right).
  \]
\item Connect the new root $\varepsilon$ to all nodes in $\$^* \otimes
  b^*$.
\item Connect the new root $b$ to all nodes in $\$^* \otimes b^+$.
\end{itemize}
Then $\D_0$ is an injectively $\omega$-automatic dag of height 3 and
we set $\H_0=\unfold(\D_0)$.  Then, we have the following:
\begin{itemize}
\item The set of roots of $\H_0$ is $\left(
    (\otimes_{n}(\{0,1\}^\omega))\otimes a^+ \right) \cup
  \{\varepsilon,b\}$.
\item For all ${\overline X}\in (2^{\N_+})^n$, $x\in \N_+$ we have:
\end{itemize}
\begin{eqnarray*}
  \H_0(w_{\overline X} \otimes a^x)  &\cong &
  r\circ
  \left(
    \begin{array}{l}
      \displaystyle
      \biguplus\{\H''(b^m) \mid m\in\N_+\}  \uplus \\
      \displaystyle
      \biguplus\{\H''(w_{\overline X} \otimes 
                  a^x\otimes a^y) \mid y\in \N_+\} 
    \end{array}
  \right)^{2^{\aleph_0}}\\
  &\cong & r\circ 
   \left( 
     \biguplus\{U''[m]\mid m\in \N_+\} \uplus
     \biguplus\{T''[\overline X,x,y]\mid y\in \N_+\}
   \right)^{\aleph_0} \\
  &\cong & T_0[\overline X,x]\\[1ex]
 \H_0(\varepsilon)&\cong & 
   r\circ \Big( \biguplus\{\H''(b^m) \mid m\in\N\} \Big)^{\aleph_0} \\
      &\cong & r\circ \Big( \biguplus\{U''[\kappa]\mid \kappa\in \N_+\cup \{\omega\}\}\Big)^{\aleph_0} \\
      &\cong & U_0[0]\\[1ex]
 \H_0(b) &\cong & 
     r\circ \Big( \biguplus\{\H''(b^m) \mid m\in \N_+\}\Big)^{\aleph_0} \\
      &\cong & 
          r\circ \Big(\biguplus\{U''[m] \mid m\in \N_+\}\Big)^{\aleph_0} \\
      &\cong & U_0[1]
\end{eqnarray*}
We now construct the forest $\H_1, \H_2, \H_3, \ldots, \H_{n}$
inductively.  For $0\leq m <n$, suppose we have obtained an
injectively $\omega$-automatic forest $\H_m=(L_m,E_m)$ as described in
the lemma.  The forest $\H_{m+1}$ is constructed as follows, where
$\alpha=m\bmod 2$:
\begin{itemize}
\item The domain of $\H_{m+1}$ is
  $\otimes_{n-m-1}(\{0,1\}^\omega)\otimes
  a^+\cup\{\varepsilon,b\}\cup(\{\$_1,\$_2\}^\omega\otimes L_m)$.
\item For $u,v\in L_m$ and $u',v'\in\{\$_1,\$_2\}^\omega$, the words
  $u'\otimes u$ and $v'\otimes v$ are connected by an edge if and only
  if $u'=v'$ and $(u,v)\in E_m$, i.e., the restriction of $\D_{m+1}$ to
  $\{\$_1,\$_2\}^\omega\otimes L_m$ is isomorphic to $\H_m^{2^{\aleph_0}}$.
\item For all $\overline X\in (2^{\N_+})^{n-m-1}$, $x\in \N_+$, connect
  the new root $w_{\overline X} \otimes a^x$ to all nodes from
  \[
  \{\$_1,\$_2\}^\omega \otimes \Big(w_{\overline X} \otimes \{0,1\}^\omega \otimes a^x \cup  b^\alpha \Big)\,.
  \]
\item Connect the new root $\varepsilon$ to all nodes from
  $\{\$_1,\$_2\}^\omega \otimes \{\varepsilon,b^\alpha\}$.
\item Connect the new root $b$ to all nodes from
  $\{\$_1,\$_2\}^\omega \otimes \{b,b^\alpha\}$.
\end{itemize}
In this way we obtain the injectively $\omega$-automatic forest
$\H_{m+1}$ such that:
\begin{itemize}
\item The set of roots of $\H_{m+1}$ is $\left(
    (\otimes_{n-m-1}(\{0,1\}^\omega))\otimes a^+ \right) \cup
  \{\varepsilon,b\}$.
\item For $\overline X\in(2^{\N_+})^{n-m-1}$ and $x\in \N_+$ we have:
  \begin{eqnarray*}\label{computations-H0}
   \H_{m+1}(w_{\overline X} \otimes a^x) &\cong &
   r\circ \Big(\biguplus\{\H_m(w_{\overline X}\otimes w_{X_{n-m}}\otimes x) \mid X_{n-m}\subseteq \N_+ \}  
\uplus \H_m(b^\alpha) \Big)^{2^{\aleph_0}} \\
    &\cong & r\circ \Big( \biguplus\{ T_m[\overline X,X_{n-m},x] \mid X_{n-m}\subseteq \N_+ \} 
 \uplus U_{m}[\alpha] \Big)^{2^{\aleph_0}} \\
      &\cong & T_{m+1}[\overline X,x] \\[1ex]
      \H_{m+1}(\varepsilon) & \cong & 
          r\circ (\H_{m}(\varepsilon)\uplus\H_m(b^\alpha))^{2^{\aleph_0}} \\
      & \cong &  r \circ (U_{m}[0]\uplus U_{m}[\alpha])^{2^{\aleph_0}} \\
      & \cong &  U_{m+1}[0]\\
      \H_{m+1}(b)& \cong & r\circ \left( \H_{m}(b^\alpha) \uplus \H_{m}(b) \right)^{2^{\aleph_0}} \\
             & \cong &  r \circ \left(U_{m}[\alpha] \uplus U_{m}[1] \right)^{2^{\aleph_0}} \\
             & \cong &  U_{m+1}[1]
   \end{eqnarray*}
 \end{itemize}
Hence we proved:

\begin{lemma}
  From each $0\leq m\leq n$, one can effectively construct an
  injectively $\omega$-automatic forest $\H_m$ such that
  \begin{itemize}
  \item the set of roots of $\H_m$ is $\displaystyle \left(
      \otimes_{n-m}(\{0,1\}^\omega)\otimes a^+ \right) \cup
    \{\varepsilon,b\}$,
  \item $\H_m(w_{\overline X} \otimes a^x)\cong T_m[\overline X,x]$
    for all $\overline X\in (2^{\N_+})^{n-m}$ and $x\in \N_+$,
  \item $\H_m(\varepsilon) \cong U_m[0]$, and
  \item $\H_m(b) \cong U_m[1]$.
  \end{itemize}
\end{lemma}
Note that $T_n[x]$ is the tree in $\H_n$ rooted at $a^x$. Hence
$T_n[x]$ is (effectively) an injectively $\omega$-automatic tree. Now
Lemma~\ref{lem:h4+.reduction} finishes the proof of
Proposition~\ref{prop:h4} and therefore of Theorem~\ref{thm:h4}.

\section{$\omega$-automatic trees of height 3}\label{sec:height3}

Recall that the isomorphism problem $\Iso(\Ti_2)$ is arithmetical by
Theorem~\ref{T-height-2} and that $\Iso(\Ti_4)$ is not by
Theorem~\ref{thm:h4}. In this section, we modify the proof of
Theorem~\ref{thm:h4} in order to show that already
$\Iso(\Ti_3)$ is not arithmetical:

\begin{theorem}\label{thm:h3}
  There exists a tree $U$ such that $\{P\in\Ti_3\mid \S(P)\cong U\}$
  is $\Pi^1_1$-hard. Hence the isomorphism problem $\Iso(\Ti_3)$ for
  injectively $\omega$-automatic trees of height 3 is $\Pi^1_1$-hard.
\end{theorem}
So let $A\subseteq\N_+$ be some set from $\Pi^1_1$. By
Proposition~\ref{prp:normal.form}, it can be written as
\[
  A=\{x\in\N_+:\forall X\;\exists y\;\forall\overline z:
  \bigwedge_{i=1}^\ell p_i(x,y,\overline z)\neq q_i(x,y,\overline z)\lor\psi_i(x,y,\overline z,X)\},
\]
where $p_i$ and $q_i$ are polynomials with coefficients in $\N$ and
$\psi_i$ is a disjunction of set constraints. As in
Section~\ref{sec:h4}, let $\varphi_{-1}(x,y,X)$ denote the subformula
starting with $\forall\overline z$, and let $\varphi_0(x,X)=\forall
y : \varphi_{-1}(x,y,X)$.  We reuse the trees $T'[X,x,y,\overline
z,z_{k+1},i]$ of height~$1$. Recall that they are all of the form
$S[e_1,e_2]$ and therefore have an even number of leaves (since the
range of the polynomial $C:\N_+^2\to\N_+$ consists of even
numbers). For $e\in\N_+$, let $S[e]$ denote the height-1 tree with
$2e+1$ leaves.

Recall that the tree $T''[X,x,y]$ encodes the set of pairs
$(e_1,e_2)\in\N_+^2$ such that $e_1\neq e_2$ or there exist $\overline
z$, $z_{k+1}$, and $i$ with $e_1=p_i(x,y,\overline z)+z_{k+1}$ and
$e_2=q_i(x,y,\overline z)+z_{k+1}$. We now modify the construction of
this tree such that, in addition, it also encodes the set $X \subseteq
\N_+$:
\begin{align*}
\widehat T[X,x,y] &= r \circ
            \left(
              \begin{array}{l}\displaystyle
               \biguplus\{ S[e]\mid e\in X\}\uplus
               \biguplus \{S[e_1,e_2] \mid e_1\neq e_2\} \uplus\\
               \displaystyle
               \biguplus \{ T'[\overline X,x,y,\overline{z},z_{k+1}i]
               \mid \overline{z}\in \N^k_+,z_{k+1} \in \N_+, 1\leq i\leq \ell\}
             \end{array}
            \right)^{\aleph_0}\\
\intertext{In a similar spirit, we define $\widehat U[\kappa,X]$ for
  $X \subseteq \N_+$ and $\kappa \in \N_+ \cup\{\omega\}$:}
\widehat U[\kappa,X] &= r \circ
            \left(
              \begin{array}{l}\displaystyle
               \biguplus\{ S[e]\mid e\in X\}\uplus
               \biguplus \{S[e_1,e_2] \mid e_1\neq e_2\} \uplus \\
               \displaystyle
               \biguplus \{ S[e,e] \mid \kappa\le e < \omega\}
             \end{array}
            \right)^{\aleph_0}
\end{align*}
Then $\widehat T[X,x,y]\cong \widehat U[\omega,Y]$ if and only if
$X=Y$ and $T''[X,x,y]\cong U''[\omega]$, i.e., if and only if $X=Y$
and $\varphi_{-1}(x,y,X)$ holds by
Lemma~\ref{lem:h3.T[X,x,y]}(b). Finally, we set
\begin{align*}
T[x] &= r \circ
        \left(
           \biguplus\{\widehat U[\kappa,X] \mid X\subseteq \N_+, \kappa\in \N_+\}
           \uplus
           \biguplus\{\widehat T[X,x,y] \mid X\subseteq \N_+, y\in\N_+\}
        \right)^{\aleph_0}\\
U &= r \circ
        \left(
           \biguplus\{\widehat U[\kappa,X] \mid X\subseteq \N_+,
           \kappa\in \N_+ \cup\{\omega\}\}
        \right)^{\aleph_0} .
\end{align*}

\begin{lemma}\label{lem:h3.tree.construct}
  Let $x\in \N_+$. Then $T[x]\cong U$ if and only if $x\in A$.
\end{lemma}

\begin{proof}
  Suppose $x\in A$. To prove $T[x]\cong U$, it suffices to show that
  any height-2 subtree of $T[x]$ is a subtree of $U$ and vice
  versa. First, let $X\subseteq\N_+$ and $y\in\N_+$. Then, by
  Lemma~\ref{lem:h3.T[X,x,y]}, there exists
  $\kappa\in\N_+\cup\{\omega\}$ with $T[X,x,y]\cong U_\kappa$ and
  therefore $\widehat T[X,x,y]\cong \widehat U[X,\kappa]$, i.e.,
  $\widehat T[X,x,y]$ appears in $U$. Secondly, let
  $X\subseteq\N_+$. From $x\in A$, we can infer that there exists some
  $y\in \N_+$ with $\forall \overline{z}: \bigwedge^{\ell}_{i=1}
  p_i(x,y,\overline{x}) \neq q_i(x,y,\overline{z}) \vee
  \psi_i(x,y,\overline{z},X)$. Then Lemma~\ref{lem:h3.T[X,x,y]}
  implies $U_\omega\cong T[X,x,y]$ and therefore
  $\widehat U[X,\omega]\cong\widehat T[X,x,y]$, i.e., $\widehat
  U[X,\omega]$ appears in $T[x]$. Thus, any height-2 subtree of $T[x]$
  is a subtree of $U$ and vice versa.

  Conversely suppose $T[x]\cong U$. Let $X\subseteq\N_+$. Then
  $\widehat U[X,\omega]$ appears in $U$ and therefore in $T[x]$. Since
  $U_\kappa\not\cong U_\omega$ for $\kappa\in\N_+$, there exists some
  $y\in\N_+$ with $U_\omega\cong T[X,x,y]$. From
  Lemma~\ref{lem:h3.T[X,x,y]} we then get $x\in A$.\qed
\end{proof}

\subsection{Injective $\omega$-automaticity}

We follow closely the procedure for $m=0$ from Section~\ref{sec:h4.automata}.

\begin{proposition}\label{prop:h3.H'}
  There exists an injectively $\omega$-automatic forest $\H'=(L',E')$
  of height 1 such that
  \begin{itemize}
  \item the set of roots equals
    $\{1,\ldots,\ell\}\otimes\{0,1\}^\omega\otimes(\otimes_{k+3}(a^+))
    \cup(b^+\otimes b^+)\cup c^+$
  \item for $1\le i\le\ell$, $X\subseteq\N_+$, $x,y,z_{k+1}\in\N_+$
    and $\overline z\in\N_+^k$, we have
    \[
    \H'(i\otimes w_X\otimes a^{(x,y,\overline z,z_{k+1})})\cong
     T'[X,x,y,\overline z,z_{k+1},i]
    \]
  \item for $e_1,e_2\in\N_+$, we have
    \[
      \H'(b^{(e_1,e_2)})\cong S[e_1,e_2]
    \]
  \item for $e\in\N_+$, we have $\H'(c^e)\cong S[e]$
  \end{itemize}
\end{proposition}

\begin{proof}
  Using Lemma~\ref{lem:polynomial.automaton} twice (with the
  polynomial $C(x_1,x_2)$ and with the polynomial $2x_1+1$) and
  Lemma~\ref{lem:h3.A[varphi]}, we can construct a B\"uchi-automaton
  $\A$ accepting
  $\{1,\ldots,\ell\}\otimes\{0,1\}^\omega\otimes(\otimes_{k+3}(a^+)) \ \cup \
  (b^+\otimes b^+) \ \cup\ c^+$ such that the number of accepting runs of
  $\A$ on the $\omega$-word $u$ equals
  \begin{enumerate}[(i)]
  \item $C(e_1,e_2)$ if $u=b^{(e_1,e_2)}$,
  \item $2e+1$ if $u=c^e$,
  \item $C(1,2)$ if $u=i\otimes w_{\overline X}\otimes
    a^{(x,y,\overline z,z_{k+1})}$ such that $\psi_i(x,y,\overline
    z,\overline X)$ holds, and
  \item $C(p_i(x,y,\overline z)+z_{k+1},q_i(x,y,\overline z)+z_{k+1})$
    if $u=i\otimes w_{\overline X}\otimes a^{(x,y,\overline
      z,z_{k+1})}$ such that $\psi_i(x,y,\overline z,\overline X)$
    does not hold.
  \end{enumerate}
  The rest of the proof is the same as that of
  Proposition~\ref{prop:H'}.\qed
\end{proof}
{}From $\H'=(L',E')$, we build an injectively $\omega$-automatic dag
$\D$ as follows:
\begin{itemize}
\item The domain of $\D$ is the set $(\{0,1\}^\omega \otimes a^+
  \otimes a^+)\cup (\{0,1\}^\omega \otimes b^*)\cup(\$^*\otimes L')$.
\item For $u,v\in L'$, the words $\$^i\otimes u$ and $\$^j\otimes v$
  are connected if and only if $i=j$ and $(u,v)\in E'$. In other
  words, the restriction of $\D$ to $\$^* \otimes L'$ is isomorphic to
  $\H'^{\aleph_0}$.
\item For all $ X\subseteq \N_+$, $x,y\in \N_+$, the new root
  $w_X\otimes a^{(x,y)}$ is connected to all nodes in
  \[
    \$^*
    \otimes
     \left(
       (\{1,\ldots,\ell\}
        \otimes w_X \otimes a^{(x,y)} \otimes (\otimes_{k+1}(a^+))
       )
       \cup \{b^{(e_1,e_2)}\mid e_1 \neq e_2\}
       \cup \{c^e\mid e\in X\}\right)\,.
  \]
\item For all $X\subseteq\N_+$, the new root $w_X\otimes\varepsilon$
  is connected to all nodes in 
  $$\$^* \otimes (\{b^{(e_1,e_2)}\mid
  e_1\neq e_2\}\cup\{c^e\mid e\in X\}).
  $$
\item For all $X\subseteq\N_+$ and $m\in \N_+$, the new root
  $w_X\otimes b^m$ is connected to all nodes in
  \[
   \$^* \otimes (\{b^{(e_1,e_2)}\mid e_1\neq e_2 \vee e_1=e_2 \geq m\}
\cup\{c^e\mid e\in X\}).
  \]
\end{itemize}
It is easily seen that $\D$ is an injectively $\omega$-automatic
dag. Let $\H''=\unfold(\D)$ which is also injectively
$\omega$-automatic by Lemma~\ref{lem:dag}. Now computations analogous
to those on page~\ref{computations-H''} (using Proposition~\ref{prop:h3.H'} instead of Proposition~\ref{prop:H'}) yield
for all $X\subseteq\N_+$ and $x,y,m\in \N_+$:
\begin{eqnarray*}
    \H''(w_X\otimes a^{(x,y)}) &
    \cong & \widehat T[X,x,y]\\
    \H''(w_X\otimes\varepsilon) &\cong & \widehat U[\omega,X] \\
    \H''(w_X\otimes b^m) &\cong & \widehat U[m,X]
\end{eqnarray*}
{}From $\H''=(L'',E'')$, we build an injectively $\omega$-automatic dag
$\D_0$ as follows:
\begin{itemize}
\item The domain of $\D_0$ equals $a^* \ \cup \ \$^*\otimes L''$.
\item For $u,v\in L''$, the words $\$^i\otimes u$ and $\$^j\otimes v$
  are connected by an edge if and only if $i=j$ and $(u,v)\in E''$. Hence the
  restriction of $\D_0$ to $\$^*\otimes L''$ is isomorphic to
  ${\H''}^{\aleph_0}$.
\item For $x\in \N_+$, connect the new root $a^x$ to all nodes in
  \[
  \$^* \otimes \left( \{0,1\}^\omega\otimes b^+ \;\cup\;\{0,1\}^\omega\otimes a^x\otimes a^+ \right).
  \]
\item Connect the new root $\varepsilon$ to all nodes in $\$^* \otimes
  \{0,1\}^\omega\otimes b^*$.
\end{itemize}
Then $\D_0$ is an injectively $\omega$-automatic dag of height 3 and
we set $\H_0=\unfold(\D_0)$.  The set of roots of $\H_0$ is
$a^*$. Calculations similar to those on page~\pageref{computations-H0}
then yield $\H_0(\varepsilon)\cong U$ and $\H_0(a^x)\cong T[x]$ for
$x\in\N_+$. Hence, $T[x]$
is (effectively) an injectively $\omega$-automatic tree. Now
Lemma~\ref{lem:h3.tree.construct} finishes the proof of the first
statement of Theorem~\ref{thm:h3}, the second follows immediately.

\begin{remark}
In our previous paper \cite{KuLiLo10}, we used an iterated application of a construction very similar to the one in this
section in order to prove that the isomorphism problem for {\em automatic trees}
of height $n \geq 2$ is hard (in fact complete) for level $\Pi^0_{2n-3}$ of the 
arithmetical hierarchy. This construction allows to handle a $\forall \exists$-quantifier block,
while increasing the height of the trees by only $1$. Unfortunately we cannot iterate
the construction of this section for $\omega$-automatic trees of height $n$ 
in order to prove a lower bound of the form $\Pi^1_{2n-5}$ for $n \geq 3$.
On the technical level, its Lemma~3.2 from \cite{KuLiLo10}, which does not hold for 
second-order formulae.
\end{remark}

\section{Upper bounds assuming CH}

We denote with {\bf CH} the continuum hypothesis: Every infinite subset
of $2^{\N}$ has either cardinality $\aleph_0$ or cardinality
$2^{\aleph_0}$. By seminal work of Cohen and G\"odel,
{\bf CH} is independent of the axiom system {\bf ZFC}.

In the following, we will identify an $\omega$-word 
$w \in \Gamma^\omega$ with the function $w : \N_+ \to \Gamma$,
(and hence with a second-order object) where $w(i) = w[i]$. 
We need the following lemma:

\begin{lemma}
  From a given B\"uchi automaton $M$ over an alphabet $\Gamma$ one can
  construct an arithmetical predicate $\acc_M(u)$ (where $u : \N_+ \to
  \Gamma$) such that: $u \in L(M)$ if and only if $\acc_M(u)$ holds.
\end{lemma}

\begin{proof}
Recall that a {\em Muller automaton} is a tuple $M=(Q, \Gamma,
\Delta, I, \F)$, where $Q$, $\Gamma$, $\Delta$, and $I$ 
have the same meaning as for B\"uchi automata but 
$\F \subseteq 2^Q$. The language $L(M)$ accepted by $M$ 
is the set of all $\omega$-words $u \in \Gamma^\omega$ for 
which there exists a run $(q_1, u[1], q_2) (q_2, u[2], q_3) \cdots$
($q_1 \in I)$ such that $\{ q \in Q \mid \exists^{\aleph_0} i :
q=q_i\} \in \F$. The Muller automaton $M$ is {\em deterministic and complete}, if 
$|I|=1$ and for all $q \in Q, a \in \Gamma$ there exists a unique $p
\in Q$ such that $(q,a,p) \in \Delta$. 

It is well known that from the given B\"uchi automaton $M$ one can
effectively construct a {\em deterministic and complete} Muller
automaton $M'=(Q, \Gamma,\Delta, \{q_0\}, \F)$ such that $L(M)=L(M')$,
see e.g.\ \cite{PerP04,Tho97handbook}.  For a given $\omega$-word $u :
\N_+ \to \Gamma$ and $i \in \N$ let $q(u,i) \in Q$ be the unique state
that is reached by $M'$ after reading the length-$i$ prefix of
$u$. Note that $q(u,i)$ is computable from $i$ (if $u$ is given as an
oracle), hence $q(u,i)$ is arithmetically definable. Now, the formula
$\acc_M(u)$ can be defined as follows:
$$
\bigvee_{A \in \F} \exists x \in \N_+ \forall y \geq x 
\bigwedge_{p\in A} \big(  q(u,y) \in A \ \wedge \ \exists z \geq y : q(u,z) = p \big)
$$
\qed  
\end{proof}

\begin{theorem} \label{thm:upper-bound}
Assuming {\bf CH}, the isomorphism problem $\Iso(\T_n)$ belongs to
$\Pi^1_{2n-4}$ for $n \geq 3$.
\end{theorem}

\begin{proof}
  Consider trees $T_i = \S(P_i)$ for $P_1,P_2 \in \T_n$.  Define the
  forest $F=(V,E)$ as $F = T_1 \uplus T_2$ For $v \in V$ let $E(v) =
  \{ w \in V : (v,w) \in E\}$ be the set of children of $v$.  Let us
  fix an $\omega$-automatic presentation $P = (\Sigma, M, M_\equiv,
  M_E)$ for $F$.  Here, $M_E$ recognizes the edge relation $E$ of $F$.
  In the following, for $u \in L(M)$ we write $F(u)$ for the subtree
  $F([u]_{R(M_\equiv)})$ rooted in the $F$-node $[u]_{R(M_\equiv)}$
  represented by the $\omega$-word $u$. Similarly, we write $E(u)$ for
  $E([u]_{R(M_\equiv)})$.  We will define a
  $\Pi^1_{2n-2k-4}$-predicate $\mathsf{iso}_k(u_1,u_2)$, where
  $u_1,u_2 \in L(M)$ are on level $k$ in $F$. This predicate expresses
  that $F(u_1) \cong F(u_2)$.

  As induction base, let $k=n-2$. Then the trees $F(u_1)$ and $F(u_2)$
  have height at most $2$. Then, as in the proof of
  Theorem~\ref{T-height-2}, we have $F(u_1) \cong F(u_2)$ if and only
  if the following holds for all
  $\kappa,\lambda\in\N\cup\{\aleph_0,2^{\aleph_0}\}$:
  \begin{align*}
    F \ \models \ & \bigg( \exists^\kappa x \in V : (([u_1], x) \in E \wedge
    \exists^\lambda y \in V : (x,y) \in E) \bigg) \ \leftrightarrow \\
    & \bigg( \exists^\kappa x \in V : (([u_2], x) \in E \wedge
    \exists^\lambda y \in V : (x,y) \in E ) \bigg)\,.
  \end{align*}
  Note that by Theorem~\ref{thm:extFOaut}, one can compute from 
  $\kappa,\lambda\in\N\cup\{\aleph_0,2^{\aleph_0}\}$ a
  B\"uchi automaton $M_{\kappa,\lambda}$
  accepting the set of convolutions of pairs of $\omega$-words
  $(u_1,u_2)$ satisfying the above formula. Hence
  $F(u_1)\cong F(u_2)$ if and only if the following arithmetical
  predicate holds:
  \[
    \forall \kappa,\lambda\in\N\cup\{\aleph_0,2^{\aleph_0}\} :\acc_{M_{\kappa,\lambda}}(u_1,u_2)\,.
  \]
Now let $0 \leq k < n-2$. We first introduce a few notations.
For a set $A$, let $\mathsf{count}(A)$ denote the set of all
countable (possibly finite) subsets of $A$.
For $\kappa  \in \N \cup \{\aleph_0\}$ we denote
with $[\kappa]$ the set $\{0,\ldots,\kappa-1\}$ (resp. $\N$)
in case $\kappa \in \N$ ($\kappa = \aleph_0$).
For a function $f : (A \times B) \to C$ and $a \in A$ let $f[a] : B \to C$
denote the function with $f[a](b) = f(a,b)$.

On an abstract level, the formula $\mathsf{iso}_{k}(u_1,u_2)$ is
 \begin{align}
   &\big(\forall x\in E(u_1)\;\exists y\in E(u_2):\mathsf{iso}_{k+1}(x,y)\big)
  \ \land    \label{eqX_1X_2(d)} \\
  & \big(\forall x\in E(u_2) \;\exists y\in
  E(u_1):\mathsf{iso}_{k+1}(x,y)\big)  \label{eqX_1X_2(e)} \ \land\\
  & \forall X_1 \in \mathsf{count}(E(u_1)) \,  \forall X_2 \in \mathsf{count}(E(u_2)) :  \label{eqX_1X_2} \\
  & \qquad\qquad  \exists x, y \in X_1 \cup X_2 : \neg
      \mathsf{iso}_{k+1}(x,y)  \; \vee  \label{eqX_1X_2(a)} \\
  & \qquad\qquad \exists x \in X_1 \cup X_2\; \exists y \in (E(u_1) \cup
      E(u_2)) \setminus (X_1 \cup X_2) : \mathsf{iso}_{k+1}(x,y) \; \vee  \label{eqX_1X_2(b)} \\
  &  \qquad\qquad |X_1| = |X_2|\,.  \label{eqX_1X_2(c)}
  \end{align}
Line \eqref{eqX_1X_2(d)} and \eqref{eqX_1X_2(e)} 
express that the children of $u_1$ and $u_2$ realize
the same isomorphism types of trees of height $n-k-1$. The rest of
the formula expresses that if a certain isomorphism type $\tau$
of height-$(n-k-1)$ trees appears countably many times below $u_1$ then
it appears with the same multiplicity below $u_2$ and vice versa.
Assuming {\bf CH} and the correctness of $\mathsf{iso}_{k+1}$, 
the formula $\mathsf{iso}_{k}(u_1,u_2)$ expresses indeed that
$F(u_1) \cong F(u_2)$.

In the above definition of $\mathsf{iso}_k(u_1,u_2)$
we actually have to fill in some details.
The countable set $X_i \in \mathsf{count}(E(u_i)) \subseteq 2^V$ of
children of $[u_i]_{R(M_\equiv)}$
(which is universally quantified in \eqref{eqX_1X_2})
can be represented as a function
$f_i : [|X_i|] \times \N \to \Sigma$ such that
the following holds:
$$
\forall j \in [|X_i|] : \acc_{M_E}(u_i\otimes f_i[j])
   \; \wedge \;
\forall j,l \in [|X_i|] : j = l  \vee \neg\acc_{M_\equiv}(f_i[j] \otimes f_i[l]) .
$$
Hence, $\forall X_i \in \mathsf{count}(E(u_i)) \cdots$ in
\eqref{eqX_1X_2} can be replaced by:
\begin{align*}
& \forall \kappa_i \in \N \cup \{\aleph_0\} \; \forall f_i :
[\kappa_i]\times \N \to \Sigma :  \\
  & \qquad\qquad (\exists j \in [\kappa_i] :   \neg\acc_{M_E}(u_i\otimes f_i[j])) \; \vee \\ 
  & \qquad\qquad (\exists j,l \in [\kappa_i] :
   j \neq l \wedge \acc_{M_\equiv}(f_i[j]\otimes f_i[l])) \; \vee \cdots.
\end{align*}
Next, the formula $\exists x, y \in X_1 \cup X_2 : \neg \mathsf{iso}_{k+1}(x,y)$
in \eqref{eqX_1X_2(a)} can be replaced by:
$$
\bigvee_{i \in \{1,2\}}  \exists j,l \in [\kappa_i] : \neg \mathsf{iso}_{k+1}(f_i[j], f_i[l]) \;\vee\;
   \exists j \in [\kappa_1] \, \exists l \in [\kappa_2] :
   \neg\mathsf{iso}_{k+1}(f_1[j], f_2[l]) .
$$
Similarly, the formula $\exists x \in X_1 \cup X_2\; \exists y \in (E(u_1) \cup
    E(u_2)) \setminus (X_1 \cup X_2) : \mathsf{iso}_{k+1}(x,y)$
in \eqref{eqX_1X_2(b)} can be replaced by
\begin{align*}
 \bigvee_{i \in \{1,2\}}  \exists j \in [\kappa_i]\; \exists v :
 \N\to\Sigma : \; &
  \mathsf{iso}_{k+1}(f_i[j],v) \ \wedge \\[-4mm]
  & (\acc_{M_E}(u_1\otimes v)
  \vee \acc_{M_E}(u_2 \otimes v) ) \ \wedge \\
& \forall l \in [\kappa_1] : \neg \acc_{M_\equiv}(f_1[l] \otimes v) \ \wedge \\
 & \forall l \in [\kappa_2] :  \neg \acc_{M_\equiv}(f_2[l] \otimes v) \ .
\end{align*}
Note that in line \eqref{eqX_1X_2(d)} and \eqref{eqX_1X_2(e)} 
we introduce a new $\forall\exists$ second-order block of quantifiers.
The same holds for the rest of the formula: We introduce two universal set quantifiers in
\eqref{eqX_1X_2} followed by the existential quantifier $\exists v :
\N\to\Sigma$ in the above formula.  Since by induction,
$\text{iso}_{k+1}$ is a $\Pi^1_{2n-2(k+1)-4}$-statement, it follows
that $\text{iso}_{k}(u_1,u_2)$ is a $\Pi^1_{2n-2k-4}$-statement.  \qed
\end{proof}
Corollary~\ref{C-2nd-order-arithmetic} and \ref{thm:upper-bound}
imply:

\begin{corollary}
  Assuming {\bf CH}, the isomorphism problem for (injectively)
  $\omega$-automatic trees of finite height is recursively equivalent
  to the second-order theory of $(\N; +, \times)$.
\end{corollary}

\begin{remark}
For the case $n=3$ we can avoid the use of {\bf CH}
in Theorem~\ref{thm:upper-bound}: Let us consider
the proof of  Theorem~\ref{thm:upper-bound} for $n = 3$.
Then, the binary relation
$\mathsf{iso}_1$ (which holds between two $\omega$-words $u,v$
in $F$ if and only if $[u]$ and $[v]$ are on level 1 and $F(u) \cong
F(v)$) is a $\Pi^0_1$-predicate. It follows that this relation is
Borel (see e.g. \cite{Kech95} for background on Borel sets).
Now let $u$ be an $\omega$-word on level $1$ in $F$.
It follows that the set of all $\omega$-words
$v$ on level 1 with $\mathsf{iso}_1(u,v)$ is again
Borel. Now, every uncountable Borel set has cardinality $2^{\aleph_0}$
(this holds even for analytic sets \cite{Kech95}). It follows that
the definition of $\mathsf{iso}_0$ in the proof of
Theorem~\ref{thm:upper-bound} is correct even without assuming {\bf CH}.
Hence, $\Iso(\T_3)$ belongs to
$\Pi^1_2$ (recall that we proved $\Pi^1_1$-hardness for this problem
in Section~\ref{sec:height3}),
this can be shown in {\bf ZFC}.
\end{remark}

\section{Open problems}

The main open problem concerns upper bounds in case we assume the
negation of the continuum hypothesis.  Assuming $\neg \mathbf{CH}$, is
the isomorphism problem for (injectively) $\omega$-automatic trees of
height $n$ still analytical? In our paper \cite{KuLiLo10} we also
proved that the isomorphism problem for automatic linear orders is not
arithmetical. This leads to the question whether our techniques for
$\omega$-automatic trees can be also used for proving lower bounds on
the isomorphism problem for $\omega$-automatic linear orders. More
specifically, one might ask whether the isomorphism problem for
$\omega$-automatic linear orders is analytical.
A more general question asks for the complexity of the isomorphism
problem for $\omega$-automatic structures in general. On the face of
it, it is an existential third-order property (since any isomorphism
has to map second-order objects to second-order objects). But it is
not clear whether it is complete for this class.

\def\cprime{$'$}

\end{document}